\documentclass{report} 

\author{Ilya Volnyansky} 
\title{} 
\date{} 
\usepackage{amsmath} 
\usepackage{amssymb} 
\usepackage{verbatim} 
\usepackage{amsthm} 
\usepackage{algorithm2e} 
\usepackage[dvips]{graphicx} 
\newtheorem{theorem}{Theorem}[chapter] 
  
\newtheorem{lemma}[theorem]{Lemma} 
 
\theoremstyle{definition} 
 
\newtheorem{definition}[theorem]{Definition} 
\theoremstyle{remark}

\newtheorem{example}{Example}[chapter]

\begin{document} 
\begin{titlepage} 
\begin{center} 

\LARGE Curse of Dimensionality in the Application of Pivot-based Indexes to the Similarity Search Problem \\ 
\vspace{1in} 
\large Ilya Volnyansky\\ 
\vspace{.1in} 
Thesis Supervisor: Vladimir Pestov\\ 
\vspace{.5in} 
Thesis Submitted to the Faculty of Graduate and Postdoctoral Studies in Partial Fulfilment of the Requirements for the Degree of Master of Science in Mathematics\footnote[1]{The M.Sc. program is a joint program with Carleton University, administered by the Ottawa-Carleton Institute of Mathematics and Statistics}\\ 
\vspace{.5in} 
Department of Mathematics and Statistics\\ 
Faculty of Science\\ 
University of Ottawa\\ 
\vspace{1in} 
\copyright Ilya Volnyansky, Ottawa, Canada, 2009 
\end{center} 
\end{titlepage} 
\chapter*{Abstract} 
In this work we study the validity of the so-called curse of dimensionality for indexing of databases for similarity search. 
We perform an asymptotic analysis, with a test model based on a sequence of metric spaces $(\Omega_d)$ from which we pick datasets $X_d$ in an i.i.d. fashion. 
We call the subscript $d$ the dimension of the space $\Omega_d$ (e.g. for $\mathbb{R}^d$ the dimension is just the usual one) and we allow the size of the dataset $n=n_d$ to be such that $d$ is superlogarithmic but subpolynomial in $n$. 

We study the asymptotic performance of pivot-based indexing schemes where the number of pivots is $o(n/d)$. 
We pick the relatively simple cost model of similarity search where we count each distance calculation as a single computation and disregard the rest. 

We demonstrate that if the spaces $\Omega_d$ exhibit the (fairly common) concentration of measure phenomenon the performance of similarity search using such indexes is asymptotically linear in $n$. That is for large enough $d$ the difference between using such an index and performing a search without an index at all is negligeable. 
Thus we confirm the curse of dimensionality in this setting.

\chapter*{Introduction} 
\label{s:introduction} 

One often hears the complaint that we live in a data rich but information poor society. 
Indeed, enabled by advances in IT hardware all sorts of data are being gathered at astonishing rates. 
The force behind this is the assumption that useful information is to be found in the heaps of data. 
However slim the ratio of useful to useless data may be, the exercise is considered worthwhile. 

As this is almost inevitably complicated, a new area of active research and a business sector were created -- those of ``data mining''. 
We will take a look at what is perhaps the most fundamental data mining problem of all: given a new piece of data, finding similar pieces of data in the pile you have accumulated already. 
This is the problem of similarity search. 
It should not be confused with exact search, the topic usually discussed in introductory computer science textbooks. 

The difference is that of finding a book in a library 20 years ago and today. 
In the past, knowing that you wanted to find a book having ``Pantagruel'' in its title meant somehow finding out, by perhaps talking to librarians (who one would assume read many books) that the author is one F. Rabelais. 
Then, you would walk down the ``fiction'' section of authors starting with R, and hopefully within a few minutes locate the title. 
Today, sitting at home and without relying on knowledgeable librarians you would start reading the first page of the book within seconds of searching for ``Pantagruel'' in a Web browser. 
This miracle is due in no small part to clever solutions to the similarity search problem called indexes. 
Indexes are structures that organize a database in such a way that fast similarity search is possible. 

Is the similarity search problem solved then? 
Let us consider a slightly different problem: we are given a photograph of a mountain landscape with no clear giveaways as to the location and yet we would like to know where it was taken. 
A person who knows those mountains will immediately tell you the approximate location, as she is able to identify the particular vegetation, rock formations and glaciers. 
But finding such a person is considerably tougher than going to your local library. 
It would be most helpful if one could submit this photograph to a search engine containing millions of tagged pictures (the Web?) to find the most similar ones. 
This way our untagged picture will obtain a tag: namely geographic information. 
That no such solution exists (yet) is a testament to the difficulty of search in high dimensions. 
Untagged pictures are composed of thousands if not millions of coloured pixels and it is not immediately obvious how to teach a computer to quickly find similar ones. 

This has come to be known as the ``curse of dimensionality'' and is the primary topic of this work. 
We aim to provide an asymptotic analysis of a class of indexes applied to high dimensional datasets of ``typical'' behaviour. 
The broad conclusion is that high dimension leads almost inevitably to unacceptably slow performance of search, akin to waiting 3 days for the web browser to tell you about {\em Gargantua and Pantagruel} (enough to make the local library suddenly competitive).

In Chapter 1, we introduce the formal mathematical setting of this search problem by defining a space of all queries and datapoints. 
A distinction is made between the database, a finite collection of objects and the potentially infinite larger space that acts as the source of ``new'' objects -- the centres of queries. 
Similarity queries are based on objects from the larger space but all that is known a priori is the database which must serve to infer patterns in the larger space. 
We provide a cost model for pivot-based indexes and setup the relationships between all the relevant variables for asymptotic analysis: primarily the database size, index size and dimension.

In Chapter 2, a related phenomenon from asymptotic geometric analysis is presented: the concentration of measure. 
It is the somewhat counterintuitive observation that high-dimensional objects appear to be small when we base our measurements on samples. 
For example, when a sphere is sampled the so-called observable diameter as a function of dimension tends to 0 very quickly. 

We present various families of spaces, which we call L\'evy families, that exhibit the more particular normal concentration of measure. 
These families are used as examples of spaces that are interesting to index yet exhibit geometric properties that make indexing hard. 
We show how these geometric properties of the larger spaces imply something reminiscent of the curse of dimensionality for pivot-based indexes. 

In Chapter 3, we introduce Statistical Learning theory as a tool to connect concentration of measure on spaces to finite datasets, which are treated as random samples from these spaces. 
Our interest is in a generalization of the theorem of Glivenko-Cantelli, due to Vapnik and Chernovenkis, about uniform convergence of sampled quantities to their true values. 
The crucial condition for the applicability of this theorem is that the class of balls in a space have a low ``capacity'': among the different capacity measures that can be used here is the VC dimension. 
If this condition is met we can make conclusions about the behaviour of any random dataset. 

The main theorem is presented in Chapter 4: that concentration of measure leads, under certain very natural conditions, to the curse of dimensionality for pivot indexes. 
That is, within our model we give a mathematically rigorous proof that asymptotically in dimension, all pivot-based indexes are not significantly better than a simple sequential scan of the database. 
We derive certain properties of the speed with which this degradation in performance takes place as well. 
The asymptotic bound is strong: the performance degrades quickly in dimension to the point that at least half of all possible queries will almost surely require a full scan of the database no matter which pivot-based is used and with which parameters (under some reasonable limits on space). 
Care is taken to present the full set of assumptions of this asymptotic analysis, as its conclusion may not hold in other cases, for example when the index is allowed a large amount of storage for pre-computation. 

Although the very results we prove demonstrate that indexing in high dimensions is often hard, in Chapter 5 we perform several experiments with pivot-based indexes on two different kinds of datasets. 
The broad conclusion is that no particular flavour of index does much better than the rest: performance quickly diminishes, which would leave some tough choices for database designers who deal with high-dimensional data. 
Perhaps a pertinent question is whether real-life datasets are ever truly high dimensional: it seems to be the opinion of some researchers that they almost never are. 
In this case the doom and gloom we prognosticate is mainly theoretical. 

\tableofcontents 

\chapter{Indexing of metric space databases for similarity search} 
A database is a collection of records with an added structure that allows the user to query, update, and delete records in a variety of ways. 
Databases are extremely pervasive in modern life: the contacts list on a phone, a bank's client records, the whole Internet, a grocery store's transactions\ldots 
In fact this general definition of a database is not restricted to computerized systems: the ancient library of Alexandria was a database as well. 

What computerized systems have allowed a veritable explosion in the size and number of databases. 
A parallel for processors is the famous Moore's law that hypothesized an exponential growth in the number of components on chips. 
The original prediction shown in Figure \ref{fig:moore} has been valid, up to a constant, for over 40 years. 
Although less well documented, it has been an accepted truth in the business community that the size of databases has been expanding exponentially as well: e.g. \cite{kimball} p.295 and \cite{hegland:1}. 

What seems to be happening is a sort of Moore's Law for the size of databases as they keep pace with the rise in processing power. 

\begin{figure*} 
\centering 
		\includegraphics[height=2in,width=3in]{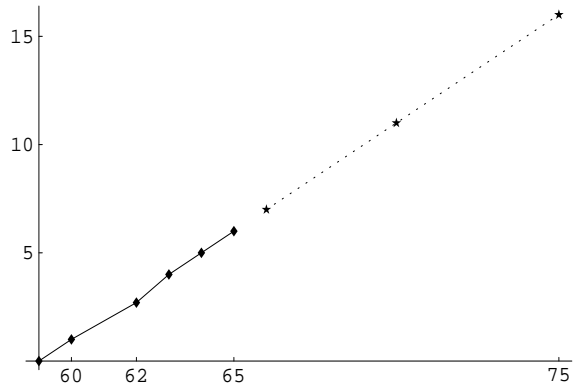} 
	\caption{\small The original Moore's Law. $\log_2$ of Components per chip vs Year} 
	\label{fig:moore} 
\end{figure*} 

At stake is the scalability of database systems, as superlinear algorithms for querying, updating and deleting are necessarily experiencing a continuous degradation of their performance. 
Thus finding more efficient algorithms is not becoming less important with rising computing speed as is sometimes suggested by non-specialists.  
Put another way, our expectations of increase in performance exceed even the astonishing pace of Moore's Law and so more clever algorithms need to be devised. 

To illustrate the main topic of this work, we will briefly summarize what searching a database typically means in real life. 
The most common databases today consist of several dozen (or more) big {\em tables}. 
An example of a 6-column database table is provided in (pardon the overloading of meaning) Table \ref{tabledatabase}. 
The two example rows are {\em records}, and typically millions of them would exist in a large company or government department. 
Customer ID is in database parlance a {\em primary key}: it is used by the system to uniquely identify the record. 

\begin{table*}[h] 
\small 
	\centering 
		\begin{tabular}{r r r r r r} 
		\textbf{Customer ID} & last name&first name&address&postal code&add date\\ 
			\hline	 
	12783 & Black &Conrad&50 Beechwood Dr.&M1F6H0&2008-05-01\\ 
25456 & Minkowski &Alice&204-190 5th Avenue &H8F1K4&2006-09-25 \\ 
		\end{tabular} 
		\label{tabledatabase} 
		\caption{a typical database table} 
\end{table*}

Perhaps the most well known function of a database is a search of its records. 
An example could be the retrieval of all customers who were entered into the database in May 2007. 
While it is possible to execute this request by looking sequentially at each record and seeing if the ``add date'' attribute matches, a more efficient solution is possible. 

Roughly speaking, in a database with $n$ records the sequential scan method proposes to look at $n$ records while the more efficient one will only do about $\log (n)$ lookups. One such method is the B-tree ( \cite{sedgewick} ), a structure that puts the values of the attribute to be searched in a tree of depth $\log(n)$ and hence exact search takes only about that many operations. 
In fact the time to perform search can be thought of as constant: even for $n=10^9$ a search query requires only 2 comparisons. \cite{sedgewick} 
 
This ``classical'' search problem is actually a very particular case of {\em similarity search}. 
To understand why other kinds of search are interesting, it is helpful to realize that a B-tree or similar structures only deal with one numeric attribute at a time. 
A structure is built on that attribute alone, which makes for efficient searches of a very particular kind: those based on that single attribute. 
Multiple B-trees would have to be built if search on other attributes is to be expected, and things get more complicated still with complex queries involving multiple attributes at the same time. 

When more dimensions are added to the mix, the problem not only gets more difficult, but no efficient solutions seem to exist. 
As multiple attributes at once are to be taken into account it is no longer possible to order all the records in a line. 
We can only say, given two records, how close they are to each other. 

A database of fingerprints, stored as an array of black and white pixels, can be used to illustrate similarity search. 
The typical query is to find all the similar fingerprints to a new sample from the street. 
This similarity is computed using some function of the pixels forming the image of the query and any candidate from the database. 
This is an expensive operation that we would like to minimize. 
In classical search we would have liked to assign a key to each of the set of fingerprints in the world. 
This way we could ask for ``fingerprint set 13747'' 

The major problem however is that the point precisely is that no two sets are exactly the same, even if taken from the same person. 
The first set will be 13747 while the second could be 78965, thus making this numbering scheme of little use -- unless we have some way of ensuring that the second set of fingerprints will be always assigned a number ``close'' to 13747. 
It is realistic to assume that police investigators are interested in the ``closest 50 matches'' which they then would like to review by hand.  
The database designer then has to provide this capability with just one piece of structural information: the function that measures the similarity between two sets of fingerprints. 

This problem therefore lends itself well to be characterized in terms of metric spaces, (alternative characterizations, without metric spaces, exist \cite{goyal:08}). 

\section{The search problem}\label{sec:search} 
Formally, the problem is framed in terms of a metric space: 

\begin{definition} 
A {\em metric space} is a set $\Omega$ equipped with function $\rho: \Omega\times\Omega\rightarrow\mathbb{R}$ s.t. 
\begin{itemize} 
\item $\rho(\omega_1, \omega_2)\geqslant 0$, $=0$ if and only if $\omega_1=\omega_2$ 
\item $\rho(\omega_1, \omega_2)=\rho(\omega_2, \omega_1)$ and 
\item $\rho(\omega_1, \omega_2)\leqslant \rho(\omega_1, \omega_3)+\rho(\omega_3, \omega_2)$ 
\end{itemize} 
\end{definition} 

The function $\rho$ is a {\em distance function} (also: {\em metric} ) and it is the main mathematical structure of $\Omega$. 
Its most important feature perhaps is the last property, known as the {\em triangle inequality}. 
Essentially it is the only tool available for inferring distances in a metric space. 

In addition the metric space is equipped with a probability measure, the definition of which we assume is known to the reader or can be looked up in a text like \cite{taylor}. 
It is a way of assigning a weight to ``measurable'' subsets in the space $\Omega$ to account for different likelihoods of a random point falling into them. 
At least in the case of a finite set it is simple enough: 
\begin{definition} 
For a finite set $X$ a {\em probability measure} on the set is a function $\mu: 2^X\rightarrow [0,1]$ s.t. 
$\sum_{x\in X}{\mu(\{x\})}=1$, so that $\forall A\subset X$, $\mu(A)=\sum_{x\in A}\mu(\{x\})$. 
\end{definition} 

What we call a {\em dataset} is a finite subset $X\subset\Omega$ equipped with the inherited metric $\rho|_X$ and normalized counting measure 
\[\mu_\# (A) = \frac{|A\cap X|}{|X|}\] 
(which is indeed a probability measure). 

There is good reason to distinguish between the space $\Omega$ and the dataset $X$ we have on hand. 
In the fingerprints database example, we cannot claim to have all known and future fingerprints already in the database. 
More likely is the opposite: every time a set of fingerprints come in for inspection, no exact matching set exists on file. 
So the new set must have come from outside $X$, yet $X$ is the space we are trying to search. 

What can $\Omega$ be in this case? All the possible fingerprint impressions, taken under all possible various conditions, of the entire human population? Past, present, future?  
The exact specification of $\Omega$ may be hard to formalize and is ultimately unimportant.  
Perhaps a parallel can be drawn with probability theory where $\Omega$ stands for sample space and $X$ for a random variable. 

Furthermore, the measure $\mu$ is as a consequence also hard to specify, but it seems reasonable to assume that in any case, some fingerprints are more likely than others to show up as queries to our database. 
Then all we can safely assume about $\mu$ is that it is some non-uniform measure, which is unknown to us. 

$X$ on the other hand is quite concrete -- we literally have a list of all the elements at our disposal. 
In addition to being a subset of $\Omega$ it is a metric subspace in the sense that the same metric $\rho$ is used to calculate distances between points of $X$. 
To simply inherit $\mu$ is problematic as we don't actually know what it is. 
Nevertheless we would like to do computations in $X$ taking probabilities into account. 
The solution is to use the counting measure, also known as the empirical measure. 
This may seem crude but it is actually a pretty good approximation of the ``real'' measure as a consequence of a well-known theorem in statistics (more on this later). 

Given this structure we can perform the following similarity queries: 
\newline Given the key $q\in\Omega$, 
\begin{itemize} 
\item Nearest neighbour: find the k closest elements in X to q 
\item Range: find all the elements in X within distance r from q 
\item Proportion (variation on the first): return k closest fraction of X to q 
\end{itemize} 
An example above was of a 50 nearest neighbour query. Supposing that such a similarity query on a fingerprints databases would rely on counting common ``features'', a range query can be ``find all records that have at least 10 features identical to this sample q''. 
What these two queries have in common is the underlying idea that they will return a small set of results that a human can then manually examine. 
Effectively the different kinds of similarity queries are closely related, and studying one sort is not a serious limitation (see \cite{chavez:1}). 

More specialized queries, e.g. finding {\em all} pairs of nearest neighbours, are approached differently both in the building and analysis of algorithms and so will not be covered. 
\section{Indexing for search} 
To answer a similarity query we can revert to the strategy of looking up every element in $x\in X$ and calculating $\rho(q,x)$. 
We will call this a linear scan, as it is clearly linear in the size of the dataset. 
However we will suppose (what is often the case) that the calculation of distances $\rho(q,x)$ is the most expensive operation \cite{chavez:1}.  

In such a setting the linear scan approach is very slow: if a single distance computation takes $1/100$th of a second, several weeks would be needed 
to traverse a database with 100 million records. Distance functions that take milliseconds and more to execute are mentioned for examples in \cite{ciaccia:98}. 
 
An {\em index} is an added structure to the database that facilitates operations like searching. 
Here we will restrict the use of the word to the scope of this work: 
\begin{definition} An index is a data structure whose aim is to speed up the execution of similarity queries on a particular dataset, typically consisting of some pre-calculated values and an algorithm. 
\end{definition} 
Let's quickly introduce a simple index, Orchard's algorithm \cite{clarkson:05}, as an example:  
\begin{example}\label{eg:orchard} 
For the $n$ points in $X$, we create an $n\times n$ matrix of distances $\rho(x,y)$: each row corresponding to all the distances to $x$, sorted in increasing order. 
To perform 1-nearest neighbour search query with centre $q$, we pick a random element $x$ and go throgh the row of distances corresponding to $x$. 
If we find $y$ s.t. $\rho(y,q)<\rho(x,q)$ we switch to the row of $y$. 

We avoid going through all the elements of $X$ by applying the following criterion: 
we stop going through the list of $y$ if the last seen element $z$ satisfies $\rho(z,y)>2\rho(y,q)$. This follows from 
\[\begin{array}{ll} 
&\rho(z,y)\leqslant\rho(y,q)+\rho(q,z)\\ 
\Rightarrow&\rho(z,q)\geqslant\rho(z,y)-\rho(y,q)\\ 
\Rightarrow&\rho(z,q)>1/2\rho(z,y)>\rho(y,q)\\ 
\end{array} 
\] 
The assumption having been applied twice in the last part. 
If we can't find $z$ such that $\rho(z,y)>2\rho(y,q)$, $y$ is returned as the answer. 
We can further improve the search by applying various strategies to avoid unnecessary lookups. 
\end{example} 

Practically speaking the reason to have an index is that the several weeks taken by a na\"\i ve approach can be reduced to a few hours, maybe even only a few minutes. 
The consequences cannot be underestimated, for, as amazing is our ability to amass data the ``killer app'' is search: a collection of 100 million sets of fingerprints is interesting, but a system capable of handling thousands of queries per day is downright useful. 

Dozens of ideas of how to build such an index were presented and new ones, having various advantages and analyzed in different ways, are invented continuously. 
Attempts at categorizing and providing a unified framework are recent: a book on the subject appeared in 2005 \cite{zezula} and the first international conference on similarity search \cite{sisap} was held in 2008. 

Several high-level views exists as to how mathematically describe the function of an index in performing similarity search. 

One is outlined in \cite{chavez:1}: an indexing scheme works by partitioning the space $\Omega$ into regions. 
In other words the space is decomposed as $\Omega=\cup_i{R_i}$ where the $R_i$ are finitely many, pairwise disjoint subsets of the space. 
A query algorithm takes advantage of this partitioning as follows. 
Through the use of some inequalities -- the triangle inequality in some form -- we {\em discard} some of the regions as it can be proved that no elements of $X$ lying in those regions can possibly be in the query results.  
The elements of $X$ in the undiscarded regions have their distances $\rho(q,x)$ computed through a linear scan to determine if they should be returned. 
As mentioned above a linear scan is nothing but a sequential lookup of each of the region's elements, so the key is to eliminate as many regions as possible. 
A high level pseudocode description is given in Algorithm \ref{alg:general}. 

\begin{algorithm*}[h] 
\SetLine 
\KwData{query, index} 
\KwResult{queryResults} 
\For{each region in index.Regions}{ 
\eIf{(NOT index.IsExcluded(region , query))}{ 
append LinearSearch(query, region) to queryResults\; 
}{} 
} 
return queryResults\; 
\caption{Use of an index for query execution: regions} 
\label{alg:general} 
\end{algorithm*} 
The triangle inequality and some variations thereof, as used in indexing to avoid sequential scan are listed in \cite{zezula}: 
\begin{itemize} 
\item Double sided triangle inequality : 
\[|\rho(\omega_1,\omega_3)-\rho(\omega_2,\omega_3)|\leqslant\rho(\omega_1,\omega_2)\leqslant\rho(\omega_1,\omega_3)+\rho(\omega_3,\omega_2)\] 
in other words, knowing distances to a third point $\omega_3$ gives us constraints on $\rho(\omega_1,\omega_2)$  
\item We only know that $r_l\leqslant\rho(\omega_1,\omega_3)\leqslant r_h$ while $\rho(\omega_2,\omega_3)$ is still known exactly. Then: 
\[\max \{\rho(\omega_2,\omega_3)-r_h, r_l-\rho(\omega_2,\omega_3),0\}\leqslant\rho(\omega_1,\omega_2)\leqslant\rho(\omega_2,\omega_3)+r_h\] 
\item Only a range is known: 
both $r_l\leqslant\rho(\omega_1,\omega_3)\leqslant r_h$  
and $r^{'}_{l}\leqslant\rho(\omega_2,\omega_3)\leqslant r^{'}_{h}$ Then 
\[\max \{r^{'}_{l}-r_h, r_l-r^{'}_{h},0\}\leqslant\rho(\omega_1,\omega_2)\leqslant r_h+r^{'}_{h}\] 
\end{itemize} 

Given these simple tools the diversity of the various indexing schemes is astonishing. 
There are tree and flat structures, structures that try to optimize inserts and deletes, trees that are deep or shallow, balanced or not, with claimed 
computational complexities from constant to exponential in $n =$ size of $X$. 
Furthermore because search is a topic of research in multiple disciplines, including for example pattern recognition \cite{devroye} where a ``nearest neighbour'' algorithm is almost equivalent to nearest neighbour search, the same solutions have been reinvented multiple times and called different names \cite{chavez:1}, \cite{clarkson:05}. 

Another reason for this multiplicity of algorithms has to do with the fact that no ``best'' algorithm has been found yet and so a number of solutions offering specific space-time tradeoffs or other features have been developed. 

It is important to note that in some situations indexing is not possible. 
Some metric spaces are so general that the distance function does not provide any usable information. 
We will first consider a trivial example, that of the 0-1 distance \cite{clarkson:05}. 
\begin{example} 
Suppose we have $(\Omega, \rho, \mu )$ such that 
\begin{equation*} 
\rho(\omega_1,\omega_2)=\left\{\begin{array}{rl} 
\ 0 & \text{if } \omega_1 =\omega_2\\ 
\ 1 & \text{if } \omega_1 \neq\omega_2\\ 
\end{array} \right. 
\end{equation*} 
Then, all query types reduce to finding an exact match. 
Furthermore no amount of storage of the distance functions between various elements of $X$ can facilitate a search query as no information can be added. 
So no {\em metric based} index works. 
\end{example} 

Of course this particular example may not seem very convincing as we already mentioned the existence of efficient solutions for exact match. 
Even in the absence of keys, we could use a hash function \cite{sedgewick} to reduce the problem to a quick binary search. 

Another, more substantive example is given by \cite{lifshits:07}. 
\begin{example} 
In this space, based on a graph generated from web mining \cite{blockeel}, the distance function is such that for any two  distinct elements $\omega_1$, $\omega_2$: 
\[1/2\leqslant\rho(\omega_1,\omega_2)\leqslant 1\] 
Then in this case inferring $\rho(\omega_1,\omega_2)$ based on knowledge of distances to a third point $\omega_3$ is not informative. 
We can verify that all the commonly used variations of the triangle inequality listed above reduce to something we already know, namely that distances lie between $1/2$ and 1. 
\end{example} 

Spaces that are impossible or at least very hard to index are by no means rare -- their high incidence is a whole subject of study. 
\section{The curse of dimensionality} 
\label{sec:curse} 
An often repeated observation is the inability of algorithms to deal with high dimensional datasets (e.g. \cite{beyer})-- a phenomenon described as the {\em curse of dimensionality}. 
Simply put, when algorithms are run on Euclidean datasets of increasing dimension, performance drops exponentially as a function of dimension. 

The concept of dimension in a general metric space is less precise. 
Clearly it has to obey our intuition in Euclidean space so for example a plane in the 10-dimensional space $\mathbb{R}^d$ is still 2-dimensional, and 
it would be desirable for a uniformly distributed ball in $\mathbb{R}^d$ to be $d$-dimensional. 
But what to do about datasets like sets of fingerprints, or movies, where an intuitive notion of dimension is harder to develop? 

One approach is to focus on the metric space properties of $\Omega$, and take advantage of already known concepts for metric spaces, such as packing numbers and $\epsilon$-nets. 
\begin{definition} 
An $\epsilon$-net $C$ of $(\Omega, \rho)$ is a set of points from $\Omega$ satisfying 
\[\cup_{c\in C}B(c,\epsilon)=\Omega\] 
where $B(c,\epsilon)=\{\omega\in\Omega | \rho(\omega, c)\leqslant\epsilon\}$ 
\end{definition} 
The size of the {\em minimal} $\epsilon$-net is called the covering number of $(\Omega, \rho)$ and denoted $\mathcal{C}(\Omega,\epsilon)$ as it is a function of $\epsilon$. 

It is not hard to convince oneself that at least in Euclidean spaces higher dimension leads to higher covering numbers. 
For the unit cube in $\mathbb{R}^d$ the covering number is on the order of $1/{\epsilon^d}$ \cite{clarkson:05}. 
This concept is extended \cite{clarkson:05}, to define the {\em Assouad dimension} (see algorithmic complexity notation later in Table \ref{tablealgorithmic}): 
\begin{definition} 
The Assouad dimension of $(\Omega, \rho)$ is the number $d$ satisfying: 
\[\sup_{\omega\in\Omega, r>0}{\mathcal{C}(B(\omega,r),\epsilon r)}=1/{\epsilon^{d+o(1)}}\] 
\end{definition} 
A small Assouad dimension is a very strong requirement: all the balls in the space have to be well behaved in the sense that they admit small covers. 
Perhaps unsurprisingly results exist that show feasibility of indexing for similarity search in case of small Assouad dimension \cite{clarkson:05}. 
However the curse of dimensionality stands: these algorithms have an exponential dependence on the Assouad dimension. 
This concept is also too complicated to compute for real datasets where $\Omega$ is unknown: it is not clear how to estimate it from $X$ alone. 

As the number of proposed concepts of ``intrinsic dimension'' for the purposes of similarity search is growing, \cite{pestov:07} outlines desirable properties we should look for. 
Included are ease of computation and definition for discrete spaces in such a way that the intrinsic dimension of $X$ is closely related to the intrinsic dimension of $\Omega$. 

An easy to compute and reasonable proposal for a dimension is mentioned in \cite{chavez:2}: 
\[\tilde{d} = \frac{\mathrm{E}(\rho(X,Y))^2}{2\mathrm{Var}(\rho(X,Y))}\] 
where $x,y\sim \mu$, the distribution of points in $\Omega$. 
Using relatively small samples, the empirical intrinsic dimension as computed for various objects in table \ref{tabledimensions} seems to give reasonable answers.  

\begin{table*}[h] 
	\centering 
		\begin{tabular}{|l |c|} 
		\hline 
		Space & calculated dimension\\ 
			\hline	 
	20-dimensional uniform unit cube & 27.6\\ 
	the NASA dataset \cite{dimacs}& 3.9\\ 
	20-dimensional uniform sphere & 20.8\\ 
	100-dimensional uniform sphere & 139.5\\ 
		\hline 
		\end{tabular} 
		\caption{Some empirical approximations of Ch\'avez intrinsic dimension } 
		\label{tabledimensions} 
\end{table*} 

This concept is based on looking at the histogram of distances 
\[\{d(q,x)\vert x\in X\}\] 
Given a query centre $q$, if the histogram of distances from $q$ to points in $X$ shows a lot of ``concentration'', this will be a hard query to process as most points will need to be checked. 
By concentrated we mean of low variance, while the mean of the distances is in the numerator to account for different scales. 
Explained another way, we would like a uniformly distributed unit cube in $\mathbb{R}^d$ have the same dimension as a uniformly distributed cube twice the size in the same space, so there is a need to normalize. 
This will often be a non-issue as we would normalize the distance function so that $\mathrm{E}(\rho(x,y))\sim 1$. 

The intrinsic dimension of \cite{chavez:2} then is an average measure of all the possible histograms taken from ``viewpoints" $q$. 
In reality both $\mu$ and $\Omega$ are either unknown or unworkable so this measure is to be estimated from $X$. 
Underlying is the assumption that datasets exhibit a certain amount of homogeneity of viewpoints, that is histograms taken from different $q$ will still look similar -- a hypothesis with some experimental validation \cite{ciaccia:98}. 
 
For ``truly'' $d$-dimensional structures in Euclidean space, e.g. uniformly distributed unit cube, this $\tilde{d}$ corresponds to $d$ (asymptotically). 
 
Using this intrinsic dimension it is possible to derive a lower bound on the number of distance computations required as a function of $\tilde{d}$ \cite{chavez:2}. 
This bound however is not very strong -- on the order of $\tilde{d}\, \ln (n)$.  
Thus only for large $\tilde{d}$ relative to $n$ is this lower bound significant. 
This leads to our next topic: how big is dimension allowed to get, in relation to $n$? 

Henceforth we will use $d$ somewhat ambiguously, referring to either the usual notion from vector space theory or one of the intrinsic dimension concepts with the understanding that they are all coincide, at least asymptotically. 
 
As mentioned in \cite{chavez:1}, for a fixed dimension $d$ and fixed $n$, we can find indexing schemes that are fast. 
The problem that we would like to analyze has to do with scaling these algorithms as {\em both} $d$ and $n$ grow. 
This requires us to make precise how fast we let these two quantities grow in relation to each other. 
This invariably leads to the use of algorithmic complexity notation, a summary of which appears in table \ref{tablealgorithmic}. 

\begin{table*}[h] 
	\centering 
		\begin{tabular}{|r |l|} 
		\hline 
		notation & definition\\ 
			\hline	 
	$f(t)=O(g(t))$ & for some $C>0$, eventually $f(t)\leq Cg(t)$\\ 
	$f(t)=o(g(t))$ & for {\em any } $C>0$, eventually $f(t)\leq Cg(t)$\\ 
	$f(t)=\Theta(g(t))$ & for some $C_1,C_2>0$, eventually $C_1g(t)\leq f(t)\leq C_2g(t)$\\ 
	$f(t)=\omega(g(t))$ & for {\em any } $C>0$, eventually $f(t)\geq Cg(t)$\\ 
		\hline 
		\end{tabular} 
		\caption{Algorithmic complexity notation} 
		\label{tablealgorithmic} 
\end{table*} 
  
An asymptotic analysis will therefore involve both: 
\[d\rightarrow\infty\] and 
\[n\rightarrow\infty\] 

We will try to argue for what the relation between $d$ and $n$ should be by going back to the more fundamental question of what is an efficient index. 
It is clear that we should be able to perform similarity queries with less time than that taken by a linear scan. 
In the language of algorithmic complexity, we require a sublinear complexity in $n$, that is 
\[\text{querytime }=o(n)\] 
where by querytime we mean the average time it takes for a similarity query to execute, time measured in distance computations. 
The average here is computed over a reasonable space of possible queries, on which we will touch later. 

Storage is also important, with at most polynomial storage allowed (but in practice even $n^2$ may be too much): 
\[\text{storage }=n^{O(1)}\] 
For our particular indexing scheme the storage is measured by the number of distances {\em stored}.
We do not make a distinction among the different ways to store a real number, in all cases it is considered as 1 unit (of cost). 
This covers a large number of indexing schemes that are essentially arrays of pre-computed distances. 

As our main concern is with asymptotic analysis it is also to specify bounds on $d$. 
We will follow an approach in the authoritative survey by \cite{indyk:1} and focus on a particular range for $d$: superlogarithmic but subpolynomial in $n$.  
Expressed using the notation,  
\begin{equation}d=\omega(\log n)\label{eq:lbN}\end{equation} 
\begin{equation}d=n^{o(1)}\end{equation} 
The reason for the lower bound \eqref{eq:lbN} is due to a case study which requires the definition of Hamming cubes. 
\begin{definition}[The Hamming Cubes $\Sigma ^d$] 
The Hamming cube of dimension $d$ is defined as the set of all binary sequences of length d, that is its elements are of the form 
\[\boldsymbol{x}=(0,1,1,0,1,\ldots,1)\] 
and the distance between two strings is just the number of elements they don't have in common divided by d: 
\[\rho(\boldsymbol{x},\boldsymbol{y})=\frac{\sum_{i=1}^{d}{|x_i-y_i|}}{d}\] 
This metric is known as the normalized Hamming distance. 
We will give the cube a uniform measure for this discussion.  
\end{definition} 

It turns out that for at least this case, when $d$ grows slowly, say $d=O(\log n)$ the entire space $\Omega$ is so small relative to the size of the dataset that all possible queries can be pre-computed and stored without breaking the polynomial storage requirement. 
The size of $\Omega=\Sigma ^d$ is just $2^d$ which becomes on the order of $n$ for sublogarithmic $d$. 
As there are on the order of $n$ possible radii, there are only $n^2$ possible queries which can be all precomputed. 
So to build a general framework for asymptotic bounds it seems necessary that $d$ grow strictly faster than $\log n$ 

As we consider algorithms that are exponential in $d$ to suffer from the curse of dimensionality, we will require querytime polynomial in $d$ (\cite{indyk:1}): 
\[\text{querytime }=d^{O(1)}\] 

This upper bound on $d$ results from the observation that if $d$ grew so fast that $n=d^{O(1)}$ a sequential scan would be polynomial in $d$. 
As nothing needs to be proven in that case, we focus on when $d$ is subpolynomial in $n$ and require an algorithm polynomial in $d$ and hence subpolynomial in $n$. 


We will adopt the view that these bounds on $d$ are a reasonable setting for the investigation of performance of various index based query algorithms. While $d$ grows fast enough to not render the problem trivial, we disregard high rates of growth for which proven examples of the ``curse'' already exist. 

Summarizing: 

\begin{table*}[h] 
	\centering 
	The goal of finding a scalable index is to find polynomial (preferably degree less than 2) $n$ storage algorithm that allows search in polynomial $d$ time.  
\end{table*} 

This stands in contrast to the {\bfseries curse of dimensionality conjecture}, whose form we borrow from \cite{indyk:1}: 

{\sffamily 
If $d=\omega(\log n)$ and $d=n^{o(1)}$, any sequence of indexes built on a sequence of datasets $X_d\subset\Sigma_d$ allowing exact nearest neighbour search in time polynomial in $d$ must use $n^{\omega(1)}$ space. 
} 

At the moment of writing a proof of above has not been found.  
We provide it here for pivot-based indexes. 
 
\section{Pivot-based indexing} 
Pivot-based indexing algorithms (for example AESA, MVPT, BKT,...see \cite{chavez:1} and \cite{zezula}) rely on a selection 
of elements from $X$ that are used as proxies for the rest of the dataset. That is, distances from all elements of $X$ to the pivot elements 
are computed and then used to cut down computations through the triangle inequality: 
 
Given pivot set $\{p_1\ldots p_k\}$, we compute the $n\times k$ array of distances  
\[\rho(x,p_i)\text{, }1\leqslant i\leqslant k\text{, }x\in X\] 
This array serves as the index.  
 
 Given a range query with radius $r$ and centre $q$, the k distances $\rho(q,p_1)\ldots \rho(q,p_k)$ are computed so that $\rho(q,x)$ can be lower-bounded as follows: 
\[\rho (q,x)\geqslant|\rho (q,p_i)-\rho (x,p_i)|\] 
since this happens for any i, we can establish: 
\[ \rho (q,x)\geqslant\sup_{1\leqslant i\leqslant k}|\rho (q,p_i)-\rho (x,p_i)| \] 
It is useful to think of a new distance function, based on the $k$ pivots: 
\[\rho_k (q,x):= \sup_{1\leqslant i\leqslant k}|\rho (q,p_i)-\rho (x,p_i)|\] 

The fact $\rho (q,x)\geqslant \rho_k (q,x)$ can be used as a condition to discard all $x$ satisfying: 
\[ \rho_k (q,x) > r \] 
Therefore the algorithm consists of checking this condition, and if it is not satisfied, performing (the expensive) distance calculation to verify if   
\[ \rho (q,x) > r \] 
Only if it is again not true do we know that the point should be returned in the query. 
This process is described in Algorithm \ref{alg:pivots}: we call the new distance function $\rho_k$ as $index.distanceK$ to emphasize that is a function belonging to the index. 

\begin{algorithm*}[h] 
\SetLine 
\KwData{query, index} 
\KwResult{queryResults} 
\For{each point in dataSet}{ 
\eIf{index.distanceK(point , query.center) $<$ query.radius}{ 
\eIf{ distance(point , query.center) $<$ query.radius}{ 
append point to queryResults 
}{} 
}{} 
} 
return queryResults\; 
\caption{Querying a pivot-based index} 
\label{alg:pivots} 
\end{algorithm*} 

We will focus on range queries with pivot-based algorithms chiefly because they are easier to execute. 
At least in theory k-nearest neighbour queries can always be simulated by a range query with the radius set to the distance to the $k$th neighbour \cite{zezula}. 

As the iteration Algorithm \ref{alg:pivots} is happening on all the points of the dataset it may appear as this algorithm does not fit the framework of ``regions''. 
But it can always be considerate a degenerate case where the regions consist of singletons of points in the dataset plus the rest: 
\[\Omega=\left(\cup_{x\in X}{\{x\}}\right)\cup\left(\Omega\backslash X\right)\] 

This differs at least on a theoretical level from the decomposition presented in \cite{chavez:1} where an equivalence relation on the points of $\Omega$ is proposed: 
\[\omega_1\sim\omega_2\iff \forall 1\leqslant i\leqslant k, \rho(\omega_1,p_i)=\rho(\omega_2,p_i)\] 
This equivalence relation is then made to induce a partition of the space. 
In Euclidian space these partitions are intersections of spheres, which for even small $k$ will be single points. 

Perhaps a more useful characterization also presented in \cite{chavez:1} is to think of the pivot based indexing as sending $\Omega$ to a different space and then performing a range similarity search in the new space: 
\[(\Omega,\rho)\longrightarrow (l^{\infty}(k),l^{\infty} \text{-norm}): \omega\longmapsto (\rho(\omega,p_i))^{k}_{i=1} )\] 
The new space consists of sequences of reals of length $k$, with the max-distance also known as the $l^{\infty}$-norm. 
This is akin to our musing at the beginning of the chapter where we admitted that having a function that sends every set of fingerprints to a number could be useful if the function had properties that allowed us to avoid a sequential scan of the {\em original space}. 

As our cost model only counts distance computations in the original space, a range search in $l^{\infty}$ is considered free. 
That is our results stand even under the generous assumption that it takes 0 time to perform a search in $l^{\infty}$. 

We will denote by $C$ all the points of $X$ satisfying  
\[ \rho_k (q,x) > r \] 
that is all the discarded elements. 
Making $C$ large is the primary way of cutting the cost of search in the setting of distance computations as dominating cost.  
Of course we can achieve this trivially with a very large number of pivots. This will defeat the purpose however as 
\[\text{Cost of range search }= k+|X \backslash C|\] 

The most often used solution is to keep adding pivots as long as it is found experimentally to decrease the cost of search. 
If $k$ is small, on the order of $\log n$ (as often space limitations require), 
the most important component of cost becomes the size of $C$ and this is where the choice of pivots would seem to matter. 
Various approaches to pivot selection have been investigated in \cite{bustos}. 
The empirical results seem to suggest that a moderate reduction in the number of distance  
computations can be achieved, although the relative improvement drops with increasing dimension.  

There are numerous refinements on this basic approach to pivot-based indexes, but the underlying idea of using the triangle inequality together with the pre-computed distances is the same. Moreover \cite{chavez:1} argues that pivot-based indexes are one of only two types of metric space indexing algorithms, the other type being also closely related. Therefore investigating this barebones pivot index can be thought of as representative of a large number of actual implementations with the unecessary complications removed. 

To recap, we are hoping that a judicious choice of the pivots (in particular their number $k$) will result in an {\em average} $C$ that is big, preferably on the order of 99\% of $n$ (the size of $X$). 
Better yet would be to guarantee that $X \backslash C$ is no more than some fixed number, say 1000, irrespective of size of $n$. 
This way only the remaining elements will have to be totally searched -- which will produce an efficient algorithm as long as we keep $k$ reasonably small. 

In situations involving the concentration of measure phenomenon this scenario cannot happen. In fact we will show that the exact opposite takes place. 
The set $C$ will almost certainly be small, and most of the dataset will have to be exhaustively searched. 
\section{Approximate Search} 
A related problem to similarity search is approximate similarity search \cite{zezula}. 
The approximate version of say nearest neighbour search only requires that the element returned be within $1+\epsilon$ distance of the ``true'' result: 

\begin{definition}{Approximate Nearest Neighbour Search.} 
Fix $\epsilon$, $\eta >0$ Let 
\[\rho_{\text{NN}}(q)= \inf \left\{\rho(q,x)|x\in X, x\neq q\right\}\] 
represent the distance from $q$ to its nearest neighbour in $X$. 
Then an approximate nearest neighbour of $q$ in $X$ is any element $\tilde{x}$ satisfying 
\[\rho(q,\tilde{x})\leqslant (1+\epsilon) \rho_{\text{NN}}(q)\] 
with probability at least $1-\eta$. An aproximate nearest neighbour search query asks for any such $\tilde{x}$, with apriori set confidence factor $\eta$. 
\end{definition} 

There are some indications that approximate search is more efficient \cite{zezula},. 
However as pointed out in (e.g. \cite{lifshits:07}, \cite{beyer}, \cite{pestov:00}),  due to the concentration that many spaces exhibit, almost all points in a typical high dimensional dataset lie at about the same distance to $q$ so approximate search is of limited usefulness. 
We will look close at concentration in the next chapter. 
\chapter{The concentration of measure phenomenon}\label{ch:concentration} 
High dimensional spaces pose a problem for algorithms:  
dimensionality appears to affect \cite{donoho} whole classes of algorithms in optimization, numerical integration and database search. 
Cost estimates (running time) of solutions depend on dimension exponentially -- something that has come to be known as the ``curse of dimensionality''. 
This term, although now liberally applied to any problem that seems to depend exponentially on dimension, was first used in \cite{bellman} to describe the unit hypercube $\mathbb{I}^d$ in $d$ dimensions. 

This space exhibits in a certain sense a growing sparseness. 
That is if we were to take a ``small hypercube'' neighborhood around a point expecting to capture a proportion $r$ of the space $\mathbb{I}^d$, the side-lengths of the neighborhood will have to be $r^{1/d}$ (\cite{friedman:1}). 
For a given proportion, say $1\%$, this means $r=0.1$ when $d=2$ yet it grows to $r=0.79$ when $d=20$. 
Meanwhile the side lengths of the whole space remain one. 

Another way of looking at this effect is to compare the volume of the unit ball and the unit hypercube. 
The volume of the unit hypercube is clearly 1, while the ball's is 
\[\sim \frac{ (2\pi)^{d/2}}{d!}\] 
a value that goes to 0 with increasing dimension. 
This observation is used to argue that most of the points in the hypercube lie ``near the edges''. 

\begin{figure} 
\centering 
	\includegraphics[height=2.5in,width=5in]{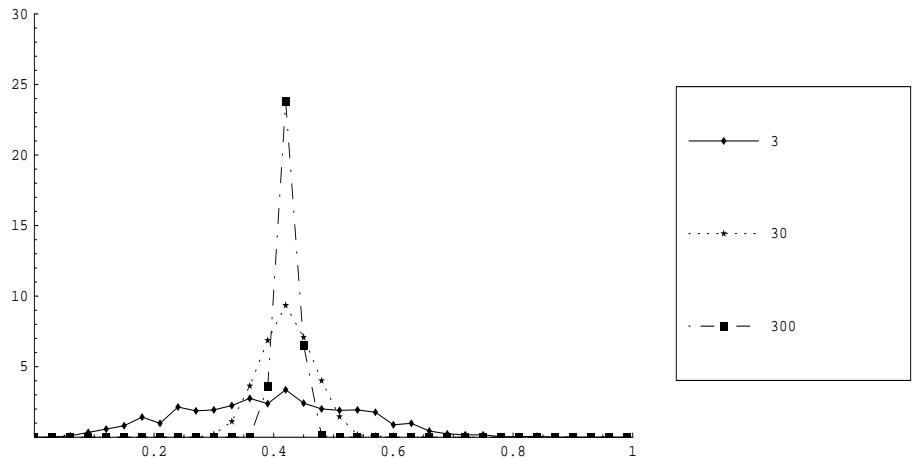} 
	\caption{Simulation of Normalized Distance between two points in the hypercube for $d=$ 3, 30, 300 } 
	\label{fig:avgdist} 
\end{figure} 

Yet another approach is to spread points uniformly and ask what is the average distance between any given two. 
In the case of the hypercube it seems that no closed form expression for this number exists, but an approximate expression is $\sqrt{d/3}$ (\cite{anderssen}). 
As there is nothing special about the centre except that it is somewhat closer to the ``average'' point, this shows that the mean distance between any two points grows at a rate of about $\sqrt{d}$ as a function of dimension. 
This is a heuristic argument often quoted for the hardness of function approximation in high dimensions (\cite{friedman:1}). 

It is legitimate to observe that the diameter of the hypercube is exactly $\sqrt{d}$ and thus a question that comes naturally to mind is why should this be called a curse of {\em dimensionality} when perhaps a more appropriate description is ``curse of {\em hugeness}''. 
After all a rectangle with side lengths $k$ and $1/k$ exhibits a similar behaviour: the average distance between two points is about $k/3$ \cite{dunbar} which will go to infinity with $k$. 
Yet it remains a low dimensional object that does not pose a problem for the aforementioned algorithms. 

So is there some effect specific to high dimension? A simulation approach is to generate (pseudo) random points on cubes of various dimensions and take the resulting histograms as approximate probability densities. 
The result presented in Figure \ref{fig:avgdist} is a more nuanced view of the distribution of distances in high dimensions: the histograms plotted show the distribution of distances for various dimensions $d$, normalized by $\sqrt{d}$. 
The average normalized distance tends to a constant, but it appears that the distribution is more {\em concentrated} with increasing dimension. 
A plot of the empirical standard deviation \cite{sheskin} shows a decrease with dimension. 
\begin{table*}[h] 
	\centering 
		\begin{tabular}{r r} 
		d & standard deviation\\ 
			\hline	 
	3 & 0.14\\ 
30 & 0.04  \\ 
300 & 0.014\\ 
3000 & 0.004\\	 
		\end{tabular} 
		\label{tablekurtosis} 
\end{table*} 

This illustrates what is sometimes called a ``benefit of dimensionality'' \cite{donoho}, namely the concentration of measure phenomenon. 

It is a well-studied topic of geometric analysis and is a much more powerful statement than the one we have made about variances. 
As a preview we shall take a look at one more set of pictures -- of the high dimensional unit sphere. 
In addition to being an appealing object it is naturally normalized with respect to distance: the maximal distance between two points remains 2 no matter the dimension. 

In order to draw a (2-D) picture of an object we have to find a way of projecting it onto flat space. 
An orthodox choice would be to take any two coordinates, say the first couple, and plot the resulting figure. 
Doing it for several different pairs will give us different views, and thus maybe the whole object will be known. 
When the $d$-dimensional sphere is sampled according to the uniform measure, and projected onto the plane by say taking the first two coordinates the result is rather peculiar: most points concentrate near the centre of the image. A simulation for various values of $d$ is provided in Figure \ref{figure:spheres}. 

\begin{figure*} 
\centering 
		\includegraphics[height=2in,width=2in]{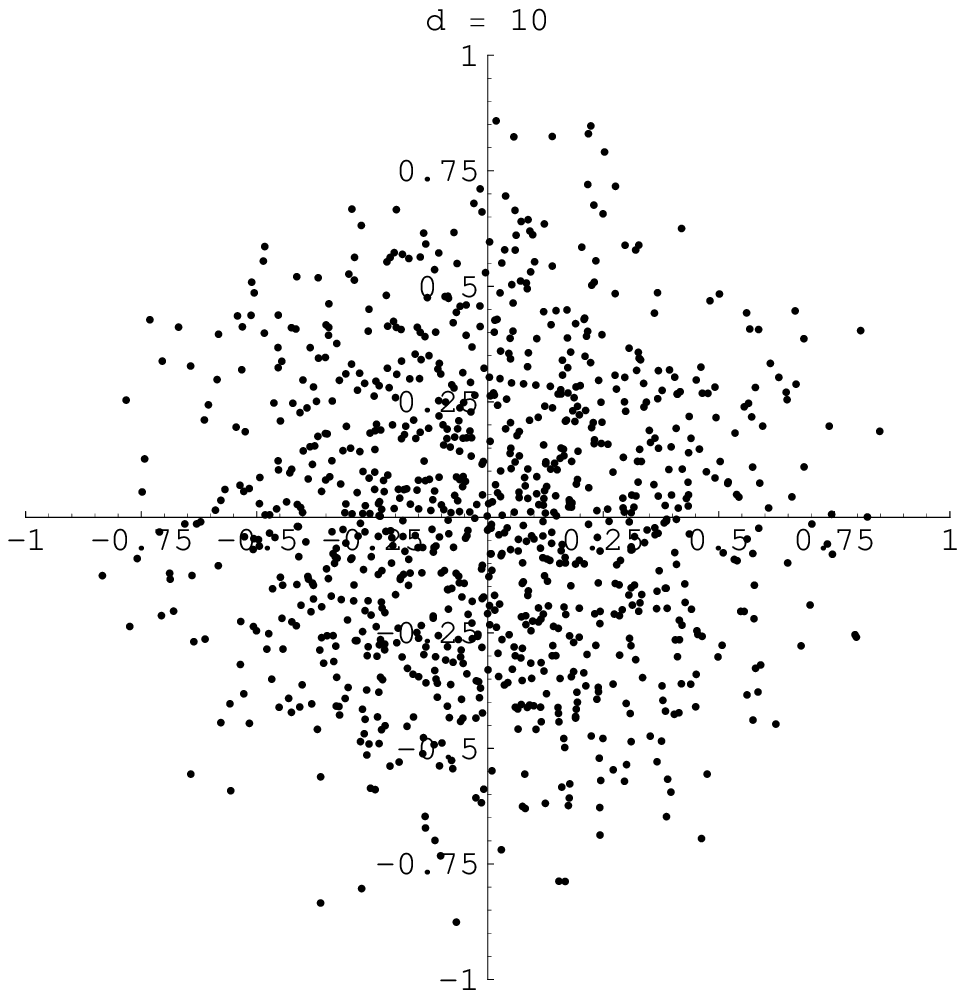} 
		\hfill 
		\includegraphics[height=2in,width=2in]{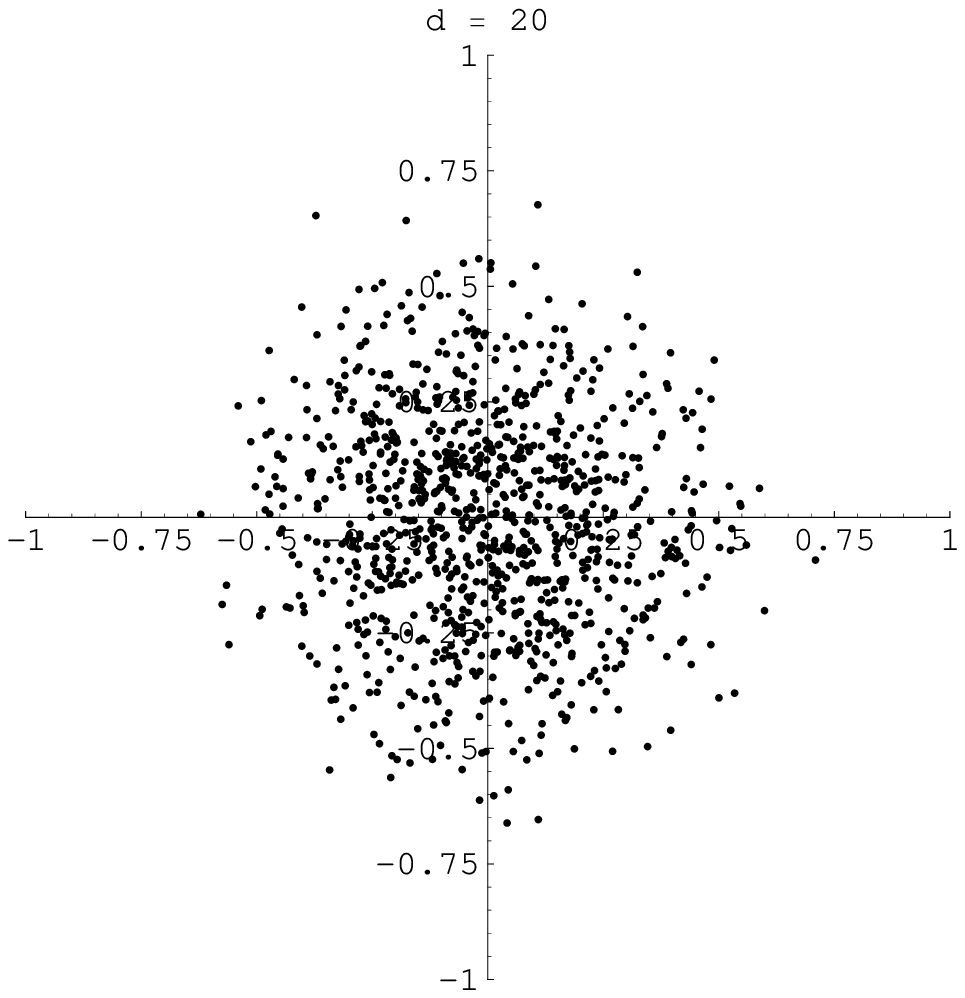} 
		\hfill 
		 
		\vspace{2ex} 
		 
		\includegraphics[height=2in,width=2in]{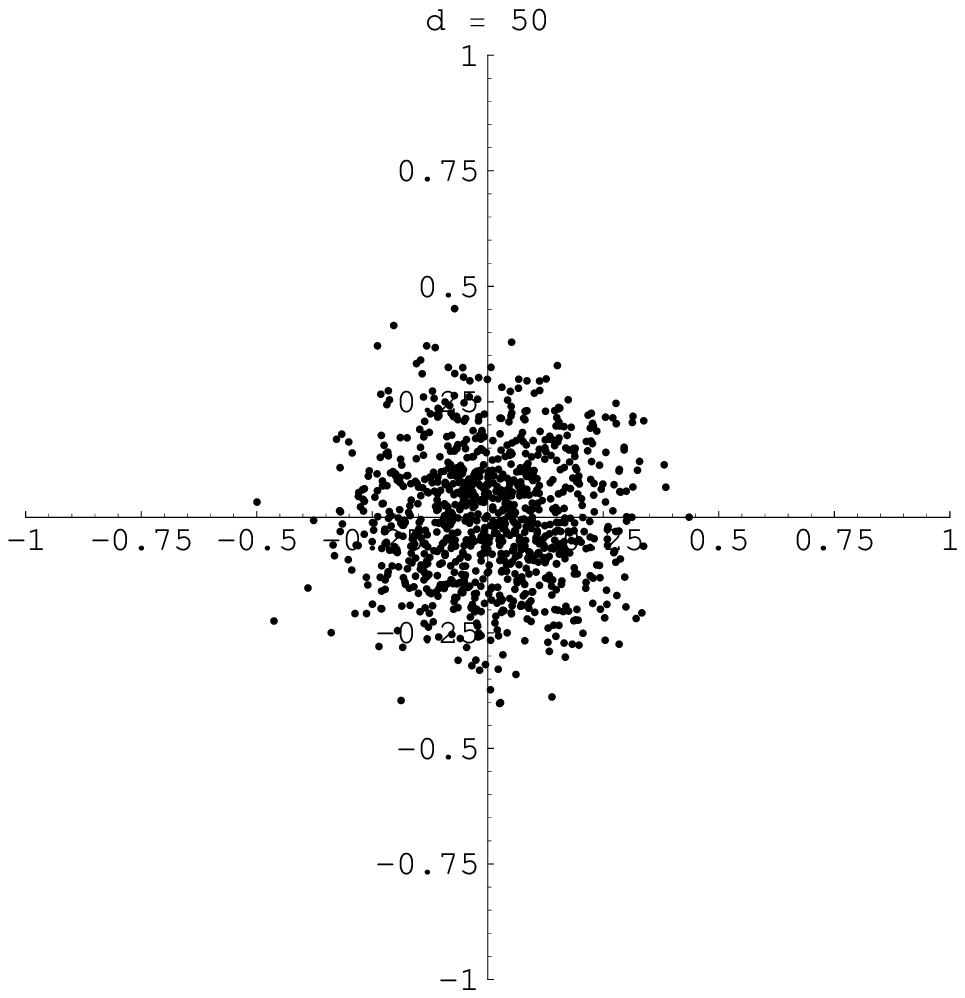} 
		\hfill 
		\includegraphics[height=2in,width=2in]{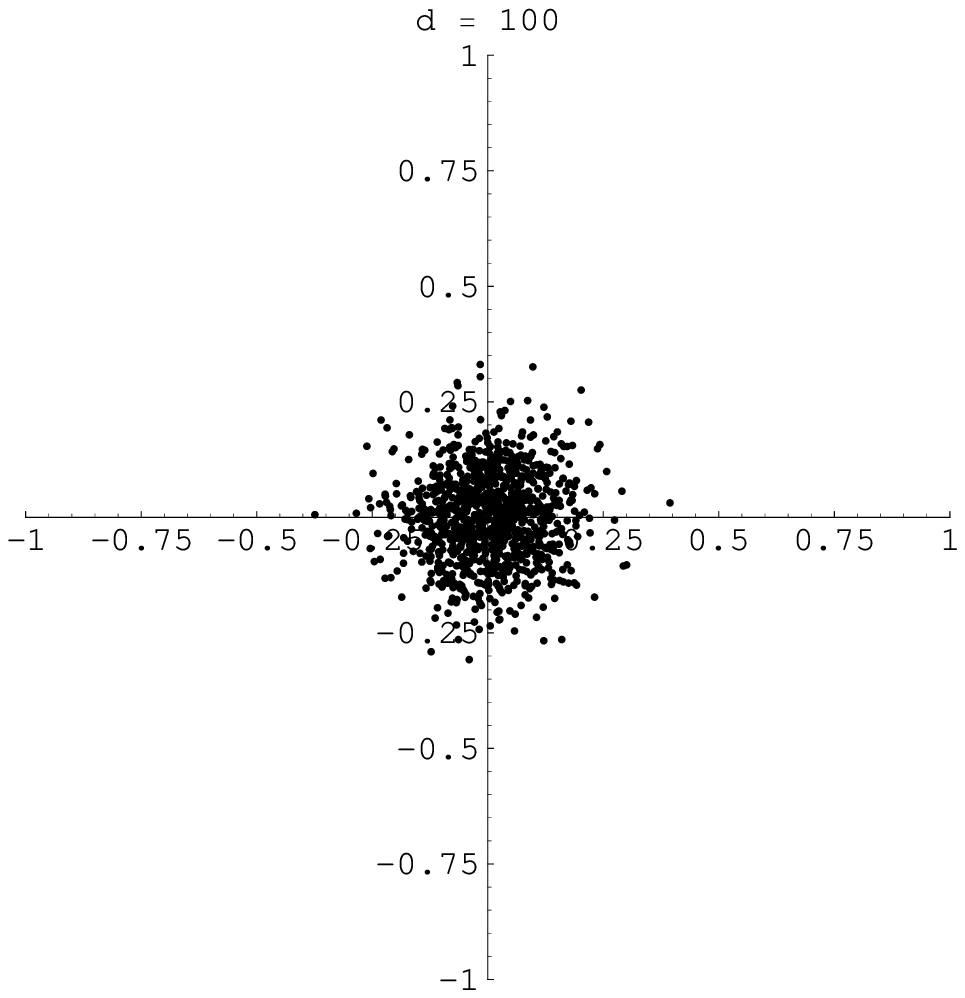} 
	\caption{Projection of uniform samples on spheres of various dimensions $d$} 
	\label{figure:spheres} 
\end{figure*} 

This happens no matter which coordinates are chosen for the projection. The 2-D picture is the same: a small core in the centre, with nearly nothing around. 
If we were to attempt to calculate the diameter based on this picture it will seem that the sphere, actually of constant diameter, is shrinking (with sample size kept constant). 

Although a bit more difficult to imagine, when taking a look at the equator of the sphere, most points will lie a short distance from it. 
Again, it doesn't matter if the equator is the standard one -- {\em any} equator will have this concentration around it. 
 
The most convenient setting for finding ``high-dimensional'' objects is $\mathbb{R}^d$ but the phenomenon of concentration is phrased in terms of measure and distances, so it can be defined on metric spaces equipped with a probability measure. As usual, whenever we take the measure of a subset of the space we will restrict the discussion to measurable sets. 

\begin{definition} 
Given a metric space $(\Omega ,\rho)$ equipped with (probability) measure $\mu$, $A_{\epsilon}$ is the $\epsilon$-neighborhood of $A\subset\Omega$, that is  
\[A_{\epsilon}=\{\omega\in\Omega|\rho(\omega,a)<\epsilon\text{ for some } a\in A\}\] 
\end{definition} 

We want to define a function $\alpha$ s.t. if $\mu(A)\geqslant 1/2$ then 
\[\mu(A_{\epsilon})\geqslant 1-\alpha ( \epsilon )\] 

In a sense we will pick the best such $\alpha$ and call it the {\em  concentration function}: 

\begin{definition} 
Given a metric space equipped with (probability) measure $(\Omega ,\rho, \mu)$ its concentration function $\alpha = \alpha_{(\Omega ,\rho, \mu)}$ is defined as 
\begin{align*} 
\alpha(0)&=1/2\\ 
\alpha(\epsilon)& = \sup\{1-\mu(A_{\epsilon})|A\subset\Omega,\mu(A)\geqslant\tfrac{1}{2}\}\quad ,\,\epsilon > 0\\ 
\end{align*} 
\end{definition} 

To put it less formally, we are trying to measure how much of the space remains after ``fat'' is added to a 
somewhat large set in the form of an $\epsilon$ neighborhood. When very little remains, we say that the concentration of measure takes place. 
Making the concept of ``little'' more precise, {\em normal }concentration of measure is considered to be taking place when $C$, $c>0$ exist such that 
\[\alpha(\epsilon)<C\text{e}^{-c d \epsilon ^2}\] 
Where $d$ is the (intrinsic) dimension. 
\section{Examples} 
\begin{example}{$\mathbb{R}^{d}$ with the Gaussian measure $\gamma$.}\label{eg:gauss} 
The Gaussian measure is defined on the completion of the Borel $\sigma$-algebra. It is the generalization of the normal probability measure on $\mathbb{R}$. For any $A$ in the above-defined $\sigma$-algebra of measurable sets, 
\[\gamma(A)= \frac{1}{{(2\pi)}^{d/2}}\int_{A}{\text{e}^{-\frac{||x||^2}{2}}\text{d}{\lambda}^d (x)}\] 
where ${\lambda}^d$ is the $d$-dimensional Lebesgue measure. 

The space $\mathbb{R}^{d}$with the measure $\gamma$ and the standard Euclidean metric coming from the $L_2$ norm 
produces a concentration function bounded as follows: 
\[\alpha(\epsilon)\leqslant\text{e}^{-\epsilon^2/2}\] 
This does not produce a normal concentration function. 
This is due to a certain stretching of the space that occurs as $d$ grows, something that is not desirable from our perspective. 
In the upcoming example of Hamming cubes we will show explicitly how a distance measure can be ``properly'' normalized so as to produce a normal concentration function. 
\end{example} 

\begin{example}{The spheres $\mathbb{S}^d$ in $\mathbb{R}^{d+1}$.} 
\begin{figure*} 
\centering 
		\includegraphics[height=2in,width=4in]{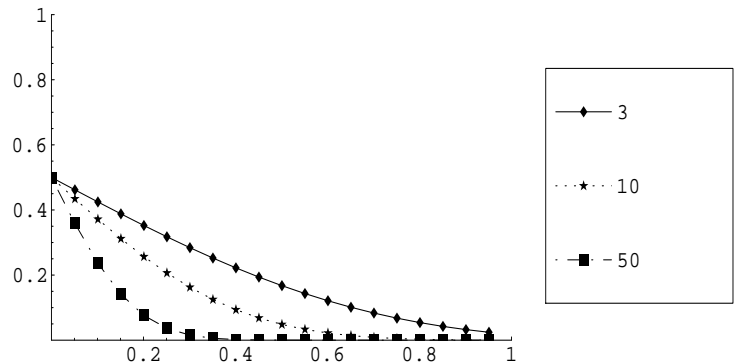} 
	\caption{The concentration functions of various spheres} 
	\label{figure:spheresConc} 
\end{figure*} 

Taken with the geodesic or Euclidian distance and the normalized invariant measure they produce a family of concentration 
functions bounded as follows \cite{L}: 
\[\alpha_d(\epsilon)\leqslant\text{e}^{-(d-1)\epsilon^2/2}\] 
In this case an exact expression for the concentration function is known \cite{benyamini} p.282, as the half-sphere, of all subsets of measure at least 1/2 will always produce the smallest $\epsilon$-neighborhood, no matter the $\epsilon$. 
The measure of this neighborhood is given by 
\[\left(\int_{-\pi/2}^{\epsilon}{ \cos^{d-2}x \text{d}x} \right)\Big/\left(\int_{-\pi/2}^{\pi/2}{ \cos^{d-2}x \text{d}x} \right)\] 
An estimation of this value can be arrived at via numeric integration. 
A plot of the resulting concentration functions, for several values of $d$, appears in Figure \ref{figure:spheresConc}.

\end{example} 
This example is particularly interesting as increasing dimension leads to increased concentration of measure phenomenon. 
\begin{definition}  
A sequence of spaces $(\Omega_d )_{d=1}^{\infty}$ a {\em normal L\'evy family} \cite{milman:86} if $C$, $c>0$ exist such that 
\[\alpha_d(\epsilon)<C\text{e}^{-c\epsilon ^2d}\] 
Thus it the same notion of what is a tight concentration function as above. 
\end{definition} 
\begin{example}{The Balls $\mathbb{B}^d$.} 
Taken with the Euclidean distance and the uniform probability measure (d-dimensional Lebesgue) form a normal L\'evy family. 
\end{example} 
\begin{example}{The Hamming Cubes $\Sigma ^d$.} 
The Hamming cubes, as defined in Section \ref{sec:curse}, with the normalized distance and uniform measure form a normal L\'evy family. 
\end{example} 
\section*{} 
The concentration of measure can be equivalently described in terms of Lipschitz functions. 

\begin{definition} 
A function $f:\Omega\rightarrow\mathbb{R}$ is {\em 1-Lipschitz }if  
\[\forall x,y\in\Omega,\,\, |f(x)-f(y)|\leqslant \rho (x,y)\] 
In general a function $f$ is p-Lipschitz if for all $x$ and $y$, $|f(x)-f(y)|\leqslant p\rho (x,y)$. 
\end{definition} 
We note that $\rho (q,.)$ the function assigning to $\omega$ its distance to $q$ is 1-Lipschitz.

In spaces that have a tight concentration $\alpha$, Lipschitz functions will be nearly constant, and one candidate for this constant is a median value. 

\begin{definition} 
\noindent A median of function $f:(\Omega,\rho,\mu)\rightarrow\mathbb{R}$ is any number $M$ satisfying: 
\[\mu\{\omega | f(\omega)\leqslant M\}\geqslant 1/2\text{ and }  \mu\{\omega | f(\omega) \geqslant M \}\geqslant 1/2\] 
This is a slight generalization of the usual concept of median of a set of numbers, only no attempt is made to make it unique.  
Discrete functions may very well have multiple valid values for $M$. 
\end{definition} 

\begin{theorem} 
\label{thm:lip} 
For a 1-Lipschitz function f defined on space $(\Omega,\mu,\rho)$: 
\[\forall\epsilon > 0,\quad \mu\{\omega |\ |f(\omega)-M|>\epsilon\} < 2\alpha(\epsilon)\] 
\end{theorem}
\begin{proof} 
Fix $\epsilon >0$. Set 
\[A=\{\omega | f(\omega)\leqslant M\}\text{ and } B=\{\omega | f(\omega)\leqslant M +\epsilon \}\] 
then 
\begin{align*}A_{\epsilon}&= \{\omega | \rho(\omega, a)<\epsilon\text{ for some } a\in A\}\\ 
													 &\subset \{\omega |f(\omega)-f(a)\leqslant\epsilon \text{ for some } a\in A\}\\ 
													 &\subset B  
\end{align*} 
Then by definition of $\alpha$, we have  
\[\mu(B)\geqslant 1- \alpha(\epsilon)\] 
The same argument can be applied to  
\[\tilde{A}=\{\omega | f(\omega)\geqslant M\}\text{ and } \tilde{B}=\{\omega | f(\omega)\geqslant M -\epsilon \}\] 
Since the probability of the entire space is 1, we rework the well-known 
\[\mu(B\cup\tilde{B})=\mu(B)+\mu(\tilde{B})-\mu(B\cap\tilde{B})\leqslant 1\] 
into 
\[\mu(B\cap\tilde{B})\geqslant\mu(B)+\mu(B)-1\] 
thus 
\[\mu(B\cap\tilde{B})\geqslant 1- \alpha(\epsilon) + 1- \alpha(\epsilon) -1 = 1-2\alpha(\epsilon)\] 
\end{proof} 

Lipschitz functions allow us to formulate the concept of observable diameter, as illustrated by Figure \ref{figure:spheres}, more rigorously. 
\begin{definition} 
Let $\kappa>0$ be fixed. 
The $\kappa$-observable diameter of $(\Omega, \rho, \mu)$, denoted $\mathrm{obs-diam}_\kappa \Omega$ is defined as ($M_f$ is the median of $f$): 
\[\mbox{obs-diam}_\kappa \Omega = 2\inf\left\{ D>0 \vert \forall \text{ 1-Lipschitz } f, \mu\{\omega |\ |f(\omega)-M_f|>D\}\leqslant\kappa\right\}\] 
\end{definition} 
We could reformulate the concentration of measure in terms of how fast the observable diameter shrinks to 0. 

Actually calculating the exponential expressions for $\alpha$ is much harder: the sometimes complicated derivations can be found in \cite{L} and \cite{milman:86}. 
However certain details of such proofs are important to understand the role normalization of distance plays. 
In spaces where the diameter is allowed to grow, e.g. cubes of dimension $d$ as at the beginning of the chapter or Hamming cubes with  
non-normalized measure 
\[\rho(\boldsymbol{x},\boldsymbol{y})=\sum_{i=1}^{d}{|x_i-y_i|}\] 
the concentration function may end up with a bound of the form: 
\[\alpha(\epsilon)\leq C\text{e}^{-c\epsilon^2}\] 
that is, without dependence on $d$. 

For purposes of similarity search however what matters is not the absolute value of the distance, but the proportion in relation to the space. 
That is, the issue is not so much if the range query is of radius 10 or 0.1 but what proportion of the space falls in a ball with this radius. 
So in an asymptotic analysis the aim is to keep the radius constant as dimension is growing -- in ``reasonable'' spaces  this is accomplished  
through normalizing by the diameter. 
In general spaces another normalization, perhaps utilizing the expected distance between two points, is required. 

We can illustrate the normalization by looking at two versions of the ``Blowing Up Lemma'' (e.g. \cite{pestov:07b} )for the Hamming Cubes: 
\begin{theorem} 
For a 1-Lipschitz function $f:\Sigma^d\rightarrow\mathbb{R}$ with respect to the non-normalized Hamming distance, we have that 
\[\forall\epsilon > 0,\quad \mu\{\omega |\ |f(\omega)-E(f)|>\epsilon\} < 2\text{e}^{\frac{-\epsilon^2}{d}}\] 
If however the function is 1-Lipshitz w.r.t. the normalized Hamming distance, we have 
\[\forall\epsilon > 0,\quad \mu\{\omega |\ |f(\omega)-E(f)|>\epsilon\} < 2\text{e}^{-d\epsilon^2}\] 
\end{theorem} 

The second part follows from the first since a 1-Lipshitz function $f$ w.r.t. the normalized Hamming distance is $1/n$-Lipshitz w.r.t the non-normalized one. 
Therefore $nf$ is 1-Lipshitz w.r.t. the non-normalized distance and so 
\[\forall\epsilon > 0,\quad \mu\{\omega |\ |df(\omega)-dE(f)|>d\epsilon\} < 2\text{e}^{\frac{-d^2\epsilon^2}{d}}=2\text{e}^{-d\epsilon^2}\] 

We are making the broad claim that the exponential decrease in $d$ is a broad phenomenon for spaces that are properly normalized for similarity search. 
So if no such dependence on $d$ is observed as in Example \ref{eg:gauss} it may be just a matter of bad choice of distance. 
\section{Link to concentration of measure}\label{sec:link} 
We would like to demonstrate why in so many familiar spaces indexing, and in particular indexing using pivots is impossible. 
The use of concentration of measure in indexing is noted in \cite{pestov:00}.  
It relies on the observation that   
\[\rho (\cdot,p):\Omega\rightarrow\mathbb{R}:\omega\mapsto \rho (\omega,p)\] 
is 1-Lipschitz for any $p$ and in particular a pivot. Hence Theorem \ref{thm:lip} can be applied to obtain a bound on the deviation from the median $M=M_p$ of function $\rho (\cdot,p)$: 
\[\forall r > 0,\quad \mu\{\omega |\ |\rho (\omega,p)-M|>r\} < 2\alpha(r)\] 
since $p$ is a general element of $\Omega$ the statement holds individually for each pivot $p_i$. We combine these statements: 
\begin{equation*} 
\forall r > 0,\quad \mu\{\omega |\sup_{1\leqslant i\leqslant k}|\rho (\omega,p_i)-M_i|>r / 2\} < 2k\alpha(r / 2) 
\end{equation*} 
as the probability of the union can always be upperbounded, if roughly, by the sum of the probabilities. 
We note that no assumptions about independence are used: the sequence $(p_i)$ can be chosen in any way. 
We used $r/2$ so as to get rid of the $M_i$: 

\[\forall r > 0,\quad \mu\{\omega |\sup_{1\leqslant i\leqslant k}|\rho (\omega,p_i)-\rho (q,p_i)|>r\} < 2k\alpha(r / 2)\] 

Comparing this to the definition of $C$: 
\[C=\{x \vert \rho _k(q,x) > r \}\] 
it is apparent that the only difference between the set we are upperbounding and $C$ is that one is defined over all of $\Omega$ and the other, just for $X$. 
We could introduce a set  
\[\mathcal{C}=\{\omega \vert \rho _k(q,\omega) > r \}\] 
and think of $C$ as the observation of $\mathcal{C}$ under $\mu_{\# }$. 

To restate the upperbound in terms of $\mathcal{C}$,  
\[\forall r > 0,\quad \mu (\mathcal{C} )< 2k\alpha(r / 2)\] 
So in effect, assuming that $2k\alpha(r / 2)$ is small, we have (roughly): 
\begin{equation} 
\mu( \mathcal{C} )\approx 0\label{zerostar1} 
\end{equation} 
this is along the lines of \cite{pestov:00} yet we would like to find out what happens to $C$.  
The point here is that if something happens in $\Omega$ it will not necessarily hold {\em mutatis mutandis} in the dataset $X$. 
\newline The statement 
\begin{equation} 
\mu_{\# }(C)\approx 0\label{zerostar2} 
\end{equation} 
describes a situation in which much of the dataset $X$ needs to be totally searched. 
The more theoretical statement \eqref{zerostar1} refers to a situation where the query searches over all the elements of $\Omega$ in effect forcing $X=\Omega$, something we would like to avoid. 

The probability measure $\mu_{\# }$ is a function of $n$, the size of the dataset.  
In fact the dataset is a sample (i.i.d) of size $n$ from the probability metric space $(\Omega,\mu)$ and hence $\mu_{\# }=\mu_n$ is also a random variable. 
Such complications show that it is important to fully describe the variables that underly the apparently simple statements. 

If we let $q$ be fixed, equation \eqref{zerostar1} in effect states 
\[\mu( \mathcal{C}_{q,p_1\ldots p_{k(n)},r(n)})\longrightarrow 0 \text{ as } n,d\longrightarrow\infty\] 
where $n$, $d$ are related as described in section \ref{sec:curse} and
\[\mathcal{C}_{q,p_1\ldots p_{k(n)},r(n)}=\{x \vert \rho _k(q,x) > r \} \] 
like before, with all variables made explicit.
Similarly \eqref{zerostar2} states 
\[\mu_{n}( C_{q,p_1\ldots p_{k(n)},r(n)})\stackrel{P}{\longrightarrow} 0 \text{ as } n,d\longrightarrow\infty\] 
where convergence in probability is based on the sample, also known as our dataset $X$ of size $n$. 

This situation is further complicated by considering the query center $q$, the pivots $p_i$, the radius $r$ and even $k$ to be random variables. 
Here we will allow $q$ to vary over all of $\Omega$, while $r$ can be thought of as proportional to the nearest neighbour in $X$ to $q$. 
In the next section we will see that the median nearest neighbour distance in our properly normalized  L\'evy families is bounded away from 0 asymptotically. 

\section{Radius of queries in L\'evy families}\label{sec:radius} 
We have described above how we would like to normalize spaces so that the ``average'' distance between two points stays about the same. 
Here we will show why this also implies the typical radius of a query -- which here we will assume to be the distance to the nearest neighbour of query centre -- also behaves nicely. 
\begin{lemma}[M. Gromov, V.D. Milman]\cite{gromov:83}\label{l:m} 
Let $(\Omega,\rho,\mu)$ denote a metric space with measure and $\alpha$ its concentration function. 
Then if $A\subset\Omega$ is such that $\mu(A)>\alpha(\gamma)$ for some $\gamma>0$, it implies that $\mu(A_{\gamma})>1/2$. 
\end{lemma} 
\begin{proof} 
Assume not and let $B=A_{\gamma}^c$. 
Then $\mu(B)\geqslant 0$, which implies $\mu(B_{\gamma}^c)\leqslant \alpha(\gamma)$. 
But $\mu(A)\leqslant \mu(B_{\gamma}^c)$ , a contradiction. 
\end{proof} 

\begin{theorem} 
Let $(\Omega_d,\rho_d,\mu_d,X_d)_{d=1}^{\infty}$ be a sequence of metric spaces with measure, forming a L\'evy family together with i.i.d. samples $X_d$ from $\Omega_d$. 
\newline Assume that $n=n_d=|X_d|=d^{o(1)}$. 
Furthermore, if $M_d$ denotes the median value of $\left\{ \rho_d(\omega_1,\omega_2)\vert \omega_i\in\Omega_d \right\}$ we assume that  
$M_d=\Theta(1)$, that is, for some fixed $c_1,c_2>0$, $\quad\forall d, c_1<M_d<c_2$. 

Let $\rho_d^{(NN)}(\omega)$ denote the distance to the nearest neighbour of $\omega\in\Omega_d$ in $X_d$. 
Define $m_d$ to be the median of $\rho_d^{(NN)}(\omega)$. 
Then there exists some $c_3>0$ and some $D$ such that $\forall d\geqslant D$, $m_d>c_3$. 
\end{theorem} 
\begin{proof} 
Assume the conclusion fails, then without loss of generality and proceeding to subsequence if necessary, $m_d\rightarrow 0$. 
By definition of $m_d$, we know that for any $d$, 
\[\mu_d\left( \bigcup_{x\in X}{B_{m_d}(x)}\right)\geqslant \frac{1}{2}\] 
It follows that  
\[n_d \sup_{\omega\in\Omega_d}{\mu_d \left( B_{m_d}(\omega)\right)}\geqslant \frac{1}{2}\] 
and so we can find for any $d$ a point $\omega_d\in\Omega_d$ such that 
\[\mu_d\left( B_{m_d}(\omega_d)\right)\geqslant \frac{1}{2n}.\] 
If we denote by $\alpha_d$ the concentration functions of our spaces $\Omega_d$ we know by assumption the existence of $C,c>0$ s.t. 
\[\forall d, \alpha_d(\epsilon)\leqslant C\text{e}^{-c\epsilon^2d}.\] 
Hence we can find $d'$ s.t. $\alpha_{d'}(\gamma) <1/{2n_{d'}}$ and $m_{d'}<{c_1}/8$ where $\gamma = c_1/8$ as well: 
this is since eventually, 
\[C\text{e}^{-c\gamma^2d} < \frac{1}{2n_d} =\frac{1}{d^{o(1)}}\] 
Then by lemma \ref{l:m} 
\[\mu_{d'}\left( B_{m_{d'}}(\omega_{d'})\right)_{\gamma}\geqslant \frac{1}{2}\] 
It then follows that 
\[\mu_{d'}\left(\left( B_{m_{d'}}(\omega_{d'})\right)_{\gamma}\right)_{\gamma}\geqslant 1-C\text{e}^{-c\gamma^2d'} \] 
that is, since $m_{d'}+2\gamma<3c/8$, 
\[\mu_{d'}\left(B_{3c_1/8}(\omega_{d'})\right)\geqslant 1-C\text{e}^{-c\gamma^2d'} \] 
But 
\[\text{diameter}\left(B_{3c_1/8}(\omega_{d'})\right)\leqslant \frac{3c_1}{4}\] 
So in $\Omega_{d'}\times\Omega_{d'}$ the measure of the set of points $(\omega_1,\omega_2)$ for which $\rho_{d'}(\omega_1,\omega_2)<c_1$ is at least 
\[\left(1-\frac{1}{2n_{d'}}\right)^2,\] 
obviously contradicting $M_{d'}>c_1$. 

\end{proof} 

This result frees us from having to consider a radius that vanishes to 0 as $n$, $d$ go to infinity. 
With this achieved, let us recap our goal: to show that most queries are slow, i.e. what we casually referred to as equation \ref{zerostar2} takes place for most queries. 
What we in fact want is something along the lines of: 
\begin{equation}\label{eq:z2} 
\mathrm{median}_{q,p_i,r}\left(\mu_{n}( C_{q,p_1\ldots p_{k(n)},r(n)})\right)\stackrel{P}{\longrightarrow} 0 \text{ as } n,d\longrightarrow\infty, 
\end{equation} 
where the median is taken over all the queries under consideration: any $q\in\Omega_d$ and any $r$ at least as large as the distance to the nearest neighbour of $q$ in $X$. 
As well, for each $d$ and $n=n_d$ we would like to also consider all possible pivot-based index schemes (as long as $k$ is within certain ranges we will specify later). 
Why the median? The aim is to show a certain behaviour for {\em many} queries: at least half is dramatic enough. 
So far we have shown, although the proof was just sketched (and with the detail about $k$ left out) that  
\begin{equation}\label{eq:z1} 
\mathrm{median}_{q,p_i,r}\left(\mu( \mathcal{C}_{q,p_1\ldots p_{k(n)},r(n)})\right)\longrightarrow 0 \text{ as } n,d\longrightarrow\infty 
\end{equation} 
Which is fine as long as $X=\Omega_d$ and hence the selection of a small dataset from an underlying large space is not taken into account. 
A more likely situation however is of a finite $X=X_d$ and an infinite (or at least much larger) $\Omega=\Omega_d$. 
What we need is to find out when statement \ref{eq:z1} implies \ref{eq:z2}. 
To do so we will summon the powerful machinery of statistical learning theory. 

\chapter{Statistical learning theory} 
Statistical learning theory has already been used in the analysis and design of indexing algorithms \cite{kleinberg} and is a vast subject. 
Instead of concerning ourselves with the whole, we will just focus on an important part: 
the generalization of the Glivenko-Cantelli theorem due to Vapnik and Chervonenkis. 

In keeping with previous notation, $X$ or $(X)$ will denote a sample of $n$ points $X_1,X_2\ldots X_n$, sampled according to some unknown probability measure $\mu$ from a space $\Omega$. 
The sampling is independent, thus in statistical jargon $X$ is an i.i.d sample. 
Although $\mu$ is unknown, we can obtain useful approximations through $X$. 
In particular we can create a measure on the space $\Omega$ induced by $X$, the counting measure $\mu_n=\mu_n(X)$: 
\[\mu_n(A)=\frac{\left\vert A\cap X\right\vert}{\vert X \vert}\] 

In section \ref{sec:search} we also referred to it as $\mu_{\#}$. 
We will use the subscript $n$ to emphasize the sample size and that in effect we are dealing with a whole class of measures. 
As they depend on the sample, these measures are also {\em random} in the sense that they are random variables. 

To talk of a random variable approximating another one, it is necessary to describe what convergence of random variables is. 

\begin{definition}[Convergence of random variables] 
It is said that a sequence of random variables $Y_n$ converges to the random variable $Y$ {\em in probability} if  
\[\forall\epsilon >0,\quad \lim_{n\rightarrow\infty}{P\left(\vert Y_n-Y \vert > \epsilon\right)} = 0\] 
We denoted convergence in probability by $Y_n\stackrel{P}{\rightarrow}Y$. 
{\em Almost everywhere }convergence takes place if a stricter condition is met: 
\[P\left(\lim_{n\rightarrow\infty}{ Y_n}=Y \right) = 1\] 
It is denoted by $Y_n \stackrel{a.e.}{\rightarrow}Y$ 
\end{definition} 
It is not too hard to see that the a.e. convergence implies convergence in probability. 
In fact the well known Law of Large Numbers has two versions, one for each type of convergence. 
We present what is in a sense just a variation on this law, the Glivenko-Cantelli theorem. 
\begin{theorem}[Glivenko-Cantelli] 
Given sample $(X) = X_1,X_2\ldots X_d$ distributed i.i.d. according to {\em any }measure $\mu$ on $\mathbb{R}^d$ we have: 
\[\sup_{r\in\mathbb{R}} \left\vert \mu_d (-\infty , r] - \mu (-\infty , r]\right\vert \stackrel{P}{\longrightarrow} 0 \] 
\end{theorem} 

\noindent 
The convergence is taken with respect the product measure induced by the sample. 
This theorem provides a means of linking the empirical distribution 
\[F_n(r):=\mu_n (-\infty , r]\] 
and the actual distribution 
\[F(r):=\mu (-\infty , r].\] 
To paraphrase, it tells us that the empirical distribution converges ``uniformly in probability'' to the actual one. 
Incidentally, the almost sure convergence also takes place. 

We can also see this statement in terms of the empirical measures of particular subsets converging to their true measure. 
This is made clear when we restate the theorem as follows:  
\begin{equation} 
\sup_{A \in\mathcal{A}} |\ \mu_n (A) - \mu (A)\ | \stackrel{P}{\longrightarrow} 0 \label{close1} 
\end{equation} 
where 
\[\mathcal{A} = \{ (-\infty , r] | r\in\mathbb{R}\}\] 
which makes more apparent a path for extension: to generalize to other collections $\mathcal{A}$, on spaces $(\Omega,\mu)$ other than the real line. 

The question of when does \eqref{close1} hold is not trivial as for a general $\mu$ and $\mathcal{A}$, the answer is not always positive: 
\begin{example} 
Let  
\[\mathcal{A}=\{\text{countable unions of open intervals in }[0,1] \} \] 
that is, the collection of sets which are countable unions of open intervals, equipped with the inherited Lebesgue measure $\mu$. 
Then for any sample $(X) = X_1,X_2\ldots X_n$ of points in $[0,1]$, we can find $A\in\mathcal{A}$ consisting of say small disjoint intervals around the $X_i$ so that 
$\mu (A) < 0.1$ and $\mu_n (A) = 1$ (by virtue of containing all the $X_i$).  
This results in $P(\sup_{A\in\mathcal{A}} | \ \mu_n (A) - \mu (A)\ |>0.8)=1$ for any $n$ and hence no 
convergence to zero can take place. 
\end{example} 
The required convergence fails to take place because the class $\mathcal{A}$ is too rich. Of course it is also possible that with a different measure convergence will take place, but in a setting where $\mu$ is unknown it is best to make little if any assumptions about it. 
Extensions of \eqref{close1} with a focus on finding the correct restrictions of $\mathcal{A}$ is a topic explored in \cite{vapnik}. 
The bulk of the work is in finding appropriate measures of ``size'' of collections that can determine if \eqref{close1} takes place, and if so at what rate. 

To that purpose, we connect the potentially infinite collection $\mathcal{A}$ with the finite sample: 

Given $(X) = X_1,X_2\ldots X_n$ the collection $\mathcal{A}$ ``colours'' the sample as follows. Each $A\in\mathcal{A}$ will assign 1 to $X_i$ if $X_i\in A$ 
otherwise it assigns 0. Hence we get a colouring of type 
 \[0,1,0,0,1,0\]  
 which is an n-length encoding that might as well have been  
\begin{center}   
 white-black-white-white-black-white 
 
\end{center}  
  We denote by $N(X)$ the {\em number} of such different encodings, when all $A\in\mathcal{A}$ are used to colour $X$. 
Clearly $N(X)\leqslant 2^n$. What is surprising is that in many situations despite a seemingly rich $\mathcal{A}$, we have $N(X)\ll 2^n$.  
\begin{definition} 
The {\em (random) entropy} of sample $X=X^{(n)}$ is just $\ln(N(X))$, denoted by $H(X)$.  
It follows that $H(X)\leqslant n\ln 2$. 
The expected value of $H(X)$, w.r.t. the sample distribution (in effect a product of $\mu$'s) is called the {\em entropy} of size n, denoted by $H(n)$: 
\[H(n) = \mathrm{E}\left(H(X^{(n)}\right)\] 
\end{definition} 

A result in \cite{vapnik} states that \eqref{close1} is equivalent to 
\[\frac{H(n)}{n}\stackrel{n\rightarrow\infty}{\longrightarrow} 0\] 

This however is of little use if $\mu$ is unknown and does not guarantee {\em fast} (i.e. exponential) convergence. 
\begin{definition} 
The {\em growth function} is defined by 
\[G(n) = \ln\ \sup_{X}N(X)\] 
where the supremum is taken over all samples of size $n\geqslant 1$. 
It is independent of $\mu$ and choice of sample $X$.  
\end{definition} 
There are two cases to consider for an upper bound for the growth function \cite{vapnik}: 
\begin{itemize} 
\item  for all $n$, $G(n)=n\ln 2$ 
\item or, for the largest $\Delta$ such that $G(\Delta)=\Delta\ln 2$, 
\begin{equation*} 
G(n)\left\{\begin{array}{rrl} 
=& n\ln 2 & \text{if } n\leqslant \Delta\\ 
\leqslant &\Delta (1+\ln (n/\Delta)) & \text{if } n>\Delta\\ 
\end{array} \right. 
\end{equation*} 
\end{itemize} 
Which says that $G$ can be either linear in $n$ or logarithmic after some point $\Delta$. 
This $\Delta$ is the so-called {\em VC dimension} and it turns out that its existence is precisely 
a necessary and sufficient condition for \eqref{close1}. 
Terminology in the literature also deems existence of $\Delta$ the case of {\em finite} VC dimension while if no such number exists $\mathcal{A}$ is said to be of infinite VC dimension. 

So the challenge of going from \eqref{eq:z1} to \eqref{eq:z2} can be reduced to the calculation of this VC dimension for some $\mathcal{A}$. 
Since equations \eqref{eq:z1} and \eqref{eq:z2} are statements about sets $\mathcal{C}=\mathcal{C}_{q,p_1\ldots p_{k(n)},r(n)}$, we would like to estimate the VC dimension of the collection of all sets of this form, for fixed $k$: 
\begin{equation} 
\label{eq:A} 
\mathcal{A}=\mathcal{A}_k=\{\mathcal{C}_{q,p_1\ldots p_{k(n)},r(n)}|q\in\Omega , p_i\in\Omega, r>0\} 
\end{equation} 
Unfortunately this is a far from trivial exercise as it involves figuring out what possible forms can of all of these sets take. 
\section{Calculating a VC dimension} 
To show an easy example of calculating a VC dimension, we will come back to the $\mathcal{A}$ from Glvenko-Cantelli: 
\[\mathcal{A} = \{ (-\infty , r] \,|\, r\in\mathbb{R}\}\] 
The calculation in this case can be done with ``one's bare hands'': 
\newline  
For $X=X_1\in\mathbb{R}$, a sample consisting of a single observation, only two colourings have to be realized. 
The element $(-\infty,X_1-1]\in\mathcal{A}$ will paint $X_1$ as 0, while $(-\infty,X_1]$ will paint it as 1. 

This elementary observation can be phrased as ``any sample of size 1 is shattered by the class $\mathcal{A}$''. 

However if one takes a two element sample where without loss of generality $X_1 < X_2$ it is quite clear that no element of type 
$(-\infty , r]$ will contain $X_2$ but not $X_1$. In other words the colouring $0,1$ cannot be realized. 
 
Therefore if we denote an $n$-size sample by $X^{(n)}$ we have that $N(X^{(1)})=2$ and $N(X^{(2)}) < 4$ for any choice of samples. 
The biggest size of the shattered sample is then 1, i.e. the VC dimension of $\mathcal{A}$ is 1. 
 
More clever arguments (e.g. \cite{dudley:84}, \cite{vapnik}, \cite{pestov:02}) are needed to prove that: 
 
\begin{itemize} 
	\item The VC dimension of half-spaces in $\mathbb{R}^d$ is $d+1$. Recall that a general hyperplane in $\mathbb{R}^d$ is defined as 
	\[\{ x\in\mathbb{R}^d | (x,v)=b\}\] 
	for some $v\in\mathbb{R}^d$, $b\in\mathbb{R}$. Hence a general half-space is of the form 
	\[\{ x\in\mathbb{R}^d | (x,v)\geqslant b\}\] 
	with the other ``half'' specified by multiplying $v$ by $-1$. 
	\item The VC-dimension of all open (or closed) balls in $\mathbb{R}^d$ 
	\[\left\{\{ x\in\mathbb{R}^d |\ \|x-v\| < r\}\quad ,\text{ where } v\in\mathbb{R}^d \text{, }r\in\mathbb{R}\right\}\] 
	 is also $d+1$. 
	 \item axis-aligned rectangular parallelepipeds in $\mathbb{R}^d$, i.e. sets of form 
	 \[ [a_1,b_1]\times [a_2,b_2]\times\ldots\times [a_d,b_d] \] 
	 have a VC dimension of $2d$ \cite{devroye} 
\end{itemize} 
Given the existence of these results for $\mathbb{R}^d$, it is relatively easy to attempt to calculate the VC dimension of $\mathcal{A}_k$ from the previous section. 

First note that 
\[ \mathcal{C}= \{\omega : \sup_i |\ \| \omega-p_i\| - \| q-p_i\| \ | > r \} = \big( \bigcap_i \{\omega : \ |\ \| \omega-p_i\| - \| q-p_i\| \ | \leqslant r \}\big)^c \] 
This allows us to proceed through several simple steps:  
\begin{itemize} 
\item 
A set of the form  
\[\{\omega : \ |\ \| \omega-p_i\| - \| q-p_i\| \ | \leqslant r \}\] 
is a spherical shell, i.e. the intersection of one ball with the complement of a smaller ball having the same center.  
A coconut shell comes to mind as a physical example. 
We note that an intersection of shells is an intersection of sets from $\mathcal{A}\cup\mathcal{A}^c$ where $\mathcal{A}$ is the collection of all balls. 
\item 
Given a collection $\mathcal{A}$ the complement collection  
\[\mathcal{A}^c=\{A^c | A\in\mathcal{A}\}\] 
has the same VC dimension. For assume $\mathcal{A}$ shatters $X$ and take any colouring of $X$. The opposite colouring, by putting 1 instead of 0 and 0 instead of 1, 
is produced by some $A\in\mathcal{A}$. Then $A^c\in\mathcal{A}^c$ produces the original colouring. The same argument can be applied in the other direction. 
\item 
The VC dimension of balls was quoted above as $d+1$, hence the VC dimension of complements of balls is $d+1$ as well.  
The VC dimension of the {\em union} of the two collections is \[(d+1) + (d+1) +1 = 2d + 3\] 
\end{itemize} 
This is a consequence of a general result \cite{vidyasagar:03}: 
\begin{lemma} 
If a collection $\mathcal{A}$ has VC dimension $\Delta_a$ and a collection $\mathcal{B}$ has VC dimension $\Delta_b$ (i.e. both are finite), the union $\mathcal{A}\cup\mathcal{B}$ 
has VC dimension at most $\Delta_a+\Delta_b+1$ 
\end{lemma} 
We note that the above is for a union of the two {\em collections}. 
Another result, this time for intersection of sets is mentioned in \cite{blumer}: 
\begin{lemma} 
For $(\Omega,\rho)=(\mathbb{R}^d,L^2)$, an upper bound on the VC dimension of $\mathcal{A}_{\cap_k}$, composed of $k$-fold interesections of elements of $\mathcal{A}$ 
of VC dimension $\Delta$ is 
\[2\Delta k\ln (3 k)\] 
\end{lemma} 

Hence at last we can conclude that the VC dimension of $\mathcal{A}_k$ for the case $\Omega\subset\mathbb{R}^n$ is bounded by 
\begin{equation}\label{eq:vc} 
2(2d+3)(2k)\ln ((3)(2k)) = k(8d+12)\ln (6k) 
\end{equation} 
where $k$ is the number of pivots. This also allows us to conclude that convergence like in equation \eqref{close1} takes place. 
Of course we only considered the case of $\mathbb{R}^d$ with the normal Euclidian metric. 

We have already mentioned that the VC dimension of axis-aligned ``boxes'' is $2d$. 
It so happens that all balls with respect the $L^{\infty}$ metric are such boxes, so we can obtain a bound on $\mathcal{A}_k$ for this metric as well. 

Another example comes from considering the Hamming cube. 
In a similar argument to section \ref{sec:curse}, we observe that are $2^d$ points in a $d$-dimensional Hamming cube, and at most $d$ different radii, so at most $d2^d$ different balls exist. 
We know from e.g. \cite{blumer} an upper bound on the VC dimension of finite collections: 
\begin{lemma}[Finite $\mathcal{A}$] 
If the class $\mathcal{A}$ is finite, its VC dimension is bounded by $\log_2\vert\mathcal{A}\vert$. 
\end{lemma} 
Disregarding the small leftover term, the VC dimension for balls in the Hamming cube is about $d$. 

Summarizing our three examples: 
\begin{theorem}\label{thm:deltabounds} 
Let us denote by $\Delta$ the VC dimension of collection $\mathcal{A}_k$ as defined in equation \eqref{eq:A}. 
Then various upper bounds on $\Delta$, depending on the space, are as follows: 
\begin{itemize} 
\item 
For $(\mathbb{R}^d, L^2)$, $\Delta\leqslant k(8d+12)\ln (6k)$ 
\item 
For $(\mathbb{R}^d, L^\infty)$, $\Delta\leqslant k(16d+4)\ln (6k)$ 
\item 
For $(\Sigma^d,\rho)$ (normalized or not), $\Delta\leqslant k(8d+8\log_2{d}+4)\ln (6k)$ 
\end{itemize} 
\end{theorem}

The case of the general metric space is harder, for the VC dimension of balls is dependent on the ``intrinsic'' dimension of the space. 
In \cite{bousquet} various capacity measures are calculated for general balls in general metric spaces.  
Crucially, the capacity measures are stated in terms of the covering numbers of the space $(\Omega, \rho)$ which in turn relate to the Assouad dimension discussed above. 
What is lacking is a computationally feasible procedure of estimating this dimension for an arbitrary $\Omega$, through a random sample $X$. 
This underscores why the Ch\'avez intrinsic dimension is so attractive, as its easy formulation allows us to estimate it accurately with moderate sample sizes (using for example the Hoeffding inequality in section \ref{sec:conv}). 
What is lacking however is structural results either linking capacity measures to the Ch\'avez intrinsic dimension directly or somehow relate it to the Assouad dimension. 

Even if we are willing to sacrifice the generality of our $\Omega$ and stick to $\mathbb{R}^d$ (after all a lot of spaces can be made to sit even if unnaturally in {\em some} $\mathbb{R}^d$) a problem remains. 
A finite VC dimension gives us convergence, but nothing about the rate of convergence has been said so far. 
A worrying sign comes from trying out some reasonable values for $d$ and $k$, e.g. $d=20$, $k=50$. 
The bound on VC dimension then becomes 49,053 -- a much larger number than say the VC dimension of spheres of the underlying Euclidian space. 
How large does $n$ have to be for the convergence result to have any value? 
\section{Rates of convergence}\label{sec:conv} 
In addition to the convergence specified by the Law of Large Numbers various inequalities exist to describe the rate of this convergence. 
An example of that is : 
\begin{theorem}[Hoeffding inequality]\label{hoeffding} Suppose we have a sequence of i.i.d random variables $X_i\in[0,1]$, $1\leqslant i\leqslant n$. Then for all $\varepsilon >0$, 
\[P\left(\left\vert \bar{X}-\mathrm{E}\right\vert \geqslant\varepsilon\right)\leqslant 2\exp (-2 n\varepsilon^2)\] 
where $\mathrm{E}$ is the expected value of any $X_i$. 
\end{theorem} 
This is more precise than the usual statement of the Law of Large numbers: 
\[P\left(\left\vert \bar{X}-\mathrm{E}\right\vert \geqslant\varepsilon\right)\stackrel{n\rightarrow\infty}{\longrightarrow} 0\] 
We shall go back to Glivenko-Cantelli once more. In terms of $\mathcal{A}$ the theorem tells us: 
\[\text{for all }\varepsilon >0,\; P\left( \sup_{A \in\mathcal{A}} |\ \mu_n (A) - \mu (A)\ |>\varepsilon \right)\stackrel{n\rightarrow\infty}{\longrightarrow} 0\] 
With $\varepsilon >0 $ considered fixed, we can ask how fast the expression on the left goes to zero.  
A consequence of the Kolmogorov-Smirnov Law \cite{vapnik} is : 
\[ P\left( \sup_{A \in\mathcal{A}} |\ \mu_n (A) - \mu (A)\ |>\varepsilon \right)<2\exp (-2 n\varepsilon ^2)\] 
This is perhaps unsurprisingly very similar to the Hoeffding inequality. 
This convergence is what we will call exponential or fast in keeping with \cite{vapnik}. 

The extension (\cite{vapnik} p.148) to the case of any $\mathcal{A}$ of finite VC dimension $\Delta$ is as follows: 
\begin{theorem}\label{thm:conv}[Generalization of Glivenko-Cantelli]\label{thm:gkg} 
For a collection $\mathcal{A}$ of subsets of $\Omega$, of finite VC dimension $\Delta$, and any measure $\mu$ on $\Omega$, we have that for any $\varepsilon>0$, 
\[ P\left[ \sup_{A \in\mathcal{A}} |\ \mu_n (A) - \mu (A)\ |>\varepsilon \right]<4\exp \left[\left(\frac{\Delta (1+\ln (2n/\Delta))}{n} -\left(\varepsilon - \frac{1}{n}\right)^2\right) n \right].\] 
\end{theorem} 
The convergence is eventually like 
\[\exp (-\varepsilon^2 n),\] 
which is again a fast rate of convergence.  

A somewhat different form is quoted in \cite{devroye}: 
\[ P\left[ \sup_{A \in\mathcal{A}} |\ \mu_n (A) - \mu (A)\ |>\varepsilon \right]<8\exp \left(\Delta (1+\ln{n / \Delta})\right) \exp\left(\frac{-n\varepsilon^2}{32}\right) .\]

Except the assumption that $\mathcal{A}$ is of finite VC dimension no other information is used. 
In fact since no information about the measure $\mu$ is incorporated the left side can be replaced by its supremum taken over all possible probability measures on the underlying space $\Omega$. 

Depending on the specific case, tighter bounds may be possible, using other capacity concepts than the VC dimension and a priori knowledge about the measure $\mu$ \cite{vapnik}, \cite{mendelson:03}. 

As an example, we will go back to covering numbers.  
We can turn any $\mathcal{A}$ into a metric space by using the {\em exclusive or} under some measure. 
Using the counting measure, an exclusive or is just: 
\[A\oplus_n B := \mu_n{\left(A\cup B - A\cap B\right)}\] 
As a result we can talk of the minimal $\varepsilon$-net $\mathcal{B}_n (\mathcal{A}, \varepsilon)$ induced by $\mu_n$, or the minimal $\varepsilon$-net $\mathcal{B}(\mathcal{A}, \varepsilon)$ induced by $\mu$. 
If the VC dimension of $\mathcal{A}$ is known to be finite, the following bound holds (cf e.g. \cite{pestov:07b}): 
\[P\left[ \sup_{A \in\mathcal{A}} \left\vert \mu_n (A) - \mu (A)\right\vert >\varepsilon \right]<8\mathrm{E}_{\mu} \left[\mathcal{B}(\mathcal{A}, \varepsilon/8)\right] \exp\left(\frac{-n\varepsilon^2}{128}\right)\] 
Then perhaps it is unsurprising that there is a relationship between covering numbers of $\mathcal{A}$ and its VC dimension (cf e.g. \cite{pestov:07b}):  
\[\log \mathcal{B}_n (\mathcal{A}, \varepsilon) \leqslant \Delta (\log n +1 - \log \Delta)\] 
To free ourselves from the particular choice of sample (which depends on the unknown $\mu$), we can introduce a new concept 
\begin{definition} 
The {\em metric entropy} of $\mathcal{A}$ is \cite{devroye} the function of $\varepsilon$ defined as: 
\[N(\mathcal{A},\varepsilon )= \sup_n \sup_{\mu_n}\mathcal{B}_n (\mathcal{A}, \varepsilon)\] 
\end{definition} 

It then folows that  
\[P\left[ \sup_{A \in\mathcal{A}} \left\vert \mu_n (A) - \mu (A)\right\vert >\varepsilon \right]<8N(\mathcal{A},\varepsilon/8 ) \exp\left(\frac{-n\varepsilon^2}{128}\right)\] 

A natural restatement of these results is to ask how big does $n$ have to be for the expression on the left to be less than some $\eta >0$.  
Solving for $\eta$ we get  
\[n\geqslant \frac{128}{\varepsilon ^2}\left( \log N(\mathcal{A},\varepsilon/8 )+\log\frac{8}{\eta}\right) \] 
The use of some technical inequalities (cf e.g. \cite{pestov:07b}) yields a similar result in terms of VC dimension: 
\begin{equation} 
\label{eq:n} 
n\geqslant \frac{128}{\varepsilon ^2}\left(\Delta\log \frac{2\text{e}^2}{\varepsilon} +\log\frac{8}{\eta}\right).  
\end{equation} 

\chapter{Linking theoretical and empirical observations} 
In this chapter we are able to tie everything together for a conclusion about indexability. 
Essentially, we we are now able to go from the theoretical observation described by \eqref{eq:z1} to \eqref{eq:z2}. 
Let us be clear that we are not proving that for a given dataset all queries are linear for all pivot based indexes. 
There are several points to keep in mind: 
\begin{itemize} 
	\item Convergence in \eqref{eq:z2} is in probability: it is conceivable that for some rare datasets indexing is easy. 
	\item The analysis is asymptotic, i.e. for large $n$ and $d$, so we do not dispute that datasets of small dimension are indexable. 
	\item We assume certain properties of our underlying spaces $(\Omega_d)$ : concentration of measure has to take place, and the VC dimension of balls is small. We think it is not unreasonable to assume these are commonly present, but again they are probably not universal. 
\end{itemize} 
Let us walk through to the main result, on the way establishing a bound for $k$. 
From the previous section, we know that only for large values of $n$ will empirical measures be close (up to $\varepsilon$) to actual measures with high likelihood (1-$\eta$ ). 
The lower bound on $n$ then naturally depends on $\varepsilon$, $\eta$ but also on the VC dimension $\Delta$ of the collection of sets $\mathcal{A}_k$ we are trying to measure. 

Let us fix $\varepsilon$ and $\eta$, e.g. for our purposes both can be $1/2$. 
Then by pooling all constants, including $\varepsilon$ and $\eta$ but not $\Delta$ we can rewrite expression \eqref{eq:n} as: 
\begin{equation}\label{eq:nsimple}n\geqslant M_1 \Delta\end{equation} 
where $M_1>1$. 
What we would like to avoid is to have the right part of this expression grow linearly in $n$. 
We know an upper bound on $\Delta$ depends on $k$ and $d$ as established in Theorem \ref{thm:deltabounds}. As our concern is for asymptotic behaviour we will simplify this bound to $ k d \ln k$ . We will generalize to define (an) acceptable asymptotic behaviour for the VC dimension of balls: linear in $d$. 

Now we would like to be able to conclude that the database sizes we encounter are large enough for Theorem \ref{thm:gkg} to hold (that is the right side of \ref{eq:nsimple} is asymptotically below $n$). 
The number of pivots $k$ can potentially ruin the day. 
There are indexing schemes, like AESA \cite{zezula} or the above-mentioned Orchard's algorithm \ref{eg:orchard} where $k=n$. 
In such cases, what we have shown above is that samples much bigger than $n$ are needed for Theorem \ref{thm:gkg} behaviour. 
Alas (or fortunately?) our sample is all of $X$ and not more -- so for large $k$ we are not able to make the desired conclusions.  

There are good arguments as to why $n^2$ storage is not practical: in a database consisting of an enormous amount of records is there really space for storing an index that is a {\em square} of that very large number? (keeping Google and the Internet as a case in point). 
It has even been argued that under certain assumptions the optimal number of pivots is on the order of $\ln n$ and that even this number is rarely reached in practice \cite{chavez:1}. 
Therefore it would not be unrealistic to restrict the analysis to pivot size much smaller than $n$. 
For example we could require 
\[k=O(\log n)\] 

We also have to keep in mind that the query algorithm we have described requires at least $k$ distance computation so if $k=n$ the query is linear in $n$ which for our purposes is no better than a linear search. 

A similar discussion applies to dimension $d$: although it is possible to consider situations when the growth of $d$ is on the order of $n$, more reasonable are situations in which it grows somewhat slower, as argued in section \ref{sec:curse}. 

Combining $d=o(n)$ with the even stronger condition on $k$, we can conclude that: 
\[\Delta=o(n)\] and hence asymptotically we know that the right side of expression \eqref{eq:nsimple} falls (much) under $n$. 
Since the aim is only for $\Delta=o(n)$, the requirement on $k$ can be weakened somewhat: 
\[k=o(\frac{n}{d})\] 
This will be our standing assumption for $k$. 

Hence we have a result of the form: 
\begin{equation}\label{eq:close2} 
P(\sup_{\mathcal{C}}\vert \mu_{\#}(C)-\mu(\mathcal{C})\vert > \varepsilon)<\eta 
\end{equation} 
This is equivalent to saying that if \eqref{eq:z1} holds then so does \eqref{eq:z2}. 

As we have not made it precise in chapter \ref{ch:concentration} as to why \eqref{eq:z1} holds, we will come back to the concentration of measures in spaces $\Omega_d$. 
Restating \eqref{zerostar1} using an exponential concentration of measure function and applying the reasoning of Section \ref{sec:link}: 
\[\mu (\mathcal{C}) \leqslant M_2k\text{e}^{-dr^2}\] 
We will sacrifice a certain number of $\mathcal{C}$ so that $r$ can be considered a constant (see section \ref{sec:radius}): we will proceed with at least half the queries having radius $r$ above a constant independent of $d$. 
Hence the quantities that vary in $d$ are $n$ and $k$. 
Since $d$ is superlogarithmic in $n$, 
\begin{equation*} 
\begin{array}{ll} 
&\forall c>0 \text{, } d>c\log n\\ 
\Rightarrow& \forall c>0 \text{, } \exp(-d) <\exp(-c\log n)\\ 
\Rightarrow&\forall c>0 \text{, }\exp(-d)<cn\\ 
\end{array} 
\end{equation*} 

So $\text{e}^{-dr^2}=o(n)$, and hence 
\[ \mu (\mathcal{C}) = o(n)\] 
so not only does $\mu(\mathcal{C})\rightarrow 0$, we have a bound on the convergence as well. 
In fact this holds for at least half the queries simulataneously so: 
\[\mathrm{median}\sup_{\mathcal{C}}\mu (\mathcal{C})=o(n)\] 
This, combined with equation \eqref{eq:close2} gives us the main result. 
\section{Main result} 
\begin{theorem} 

Consider a sequence of metric spaces $(\Omega_d,\rho_d)$ where $d=1,2,3,\ldots$ and the VC dimension of closed balls in space $(\Omega_d,\rho_d)$ as a function of $d$ is $O(d)$. 
\newline  
Furthermore the metric spaces come equipped with  Borel probabilities $\mu_d$ such that for fixed $C,c>0$ the sequence of concentration functions of $(\Omega_d,\rho_d,\mu_d)$ satisfies \[\forall\epsilon>0,\quad\alpha_d(\epsilon)\leqslant C\text{e}^{-c\epsilon^2d}\] 
For each $d$ an i.i.d. sample $X_d$ of size $n_d$ is selected from $\Omega_d$, according to $\mu_d$. 
The sample size $n_d$ is such that $d=\omega(\log n_d)$ and $d=n_d^{o(1)}$ if $d$ is expressed as a function of $n_d$. 
We treat $X_d$ as a dataset on which to build an index for similarity search. 
The index built is a pivot index using $k$ pivots, where \[k=o(n_d/d).\] 
We fix any small $\varepsilon,\eta >0$ as desired. 
Suppose we only ask queries whose radius is equal or greater to the distance to nearest neighbour of query centre $q\in\Omega_d$ in $X_d$. 
\newline\bfseries  
Then there exists a $D$ such that for all $d\geqslant D$, the probability that {\em at least half} the queries on dataset $X_d$ take less than $(1-\varepsilon)n_d$ time is less than $\eta$. 

Furthermore, if we allow the likelihood $\eta$ to depend on $d$, we can pick $\eta_d$ so that the above holds true and 
\[\lim_{D\rightarrow\infty}{\prod_{d=D}^{\infty}(1-\eta_d)} = 1\] 
We emphasize that this result is independent of the selection of pivots. 
\end{theorem} 


We have done all the calculations to show this theorem except demonstrate the limit. 
We will nevertheless make a quick summary. 
First of all, we will refer to Theorem \ref{thm:deltabounds} and Chapter \ref{ch:concentration} to provide us examples of spaces that meet all of the above requirements. 
It must be said that these families of spaces are not peculiar counterexamples but are commonly used in indexing today -- except of course that only the low dimensional ones are successfully being indexed. 

We showed above how the results of Section \ref{sec:link} use the Concentration of Measure to deduce the ``theoretical version'' of linear querying. 
We could only show it for all range queries above some fixed $r>0$, but this is not an obstacle as in all spaces that we talk of distances are scaled appropriately. In particular the expected distance to nearest neighbour approaches a constant as $d\rightarrow\infty$. 
So by adding this requirement of minimum radius we are merely disregarding queries that return no elements except perhaps the query center itself. 
To link this to querying in a discrete structure that is a dataset, we used the previous chapter's results where the assumptions included that the VC dimension of balls is small. 
The result then is that querying in these spaces is linear for all high dimensions. 

We ignored $k$ from calculating the runtime as it becomes asymptotically insignificant with increasing dimension. 
A $k$ larger than that stipulated in the main result will matter, but may trivially give linear runtime in $n_d$ (e.g.  if $k$ is linear in $n_d$ ). 
What remains unanswered here is the behaviour for $k$ in between $o(n_d/d)$ and $\Theta(n_d)$. 

To derive the asymptotic probability, equation \ref{eq:n} can be used to obtain the lower bound on $\eta$: 

\[\eta\geqslant \exp \left( \Delta\log \left(\frac{2\text{e}^2}{\varepsilon}\right)+\log8-\frac{\varepsilon^2 n}{128}\right)\] 

which as an asymptotic function of $d$ transforms into 
\[\eta=\exp (-d^{\omega(1)})\] 

Assuming independent choices of the datasets $X_d$, and assuming that for each $d$ the probability of an event is at least $1-\eta_d$, we will estimate the quantity 

\[\prod_{d=D}^{\infty}(1-\eta_d)\] 

As $\eta_d$ goes to 0 at least as fast as $\text{e}^{-d}$, it is enough to show that 

\[\lim_{D\rightarrow\infty}{\prod_{d=D}^{\infty}(1-\text{e}^{-d})} = 1\] 

to have  

\[\lim_{D\rightarrow\infty}{\prod_{d=D}^{\infty}(1-\eta_d)} = 1\] 

as well. 

Observing \cite{ash} that for any sequence $0\leqslant\eta_d\leqslant 1$, 
\[1-\sum_{d=1}^{N}{\eta_d}\leqslant \prod_{d=1}^{N}(1-\eta_d)\leqslant \exp\left(\sum_{d=1}^{N}{-\eta_d}\right)\] 
we can extend this, for any $D$ to: 
\[1-\sum_{d=D}^{\infty}{\eta_d}\leqslant \prod_{d=D}^{\infty}(1-\eta_d)\leqslant \exp\left(\sum_{d=D}^{\infty}{-\eta_d}\right)\] 
We know that as a geometric series, 
\[\sum_{d=D}^{\infty}{\text{e}^{-d}}=\frac{\text{e}^{-d}}{1-\text{e}}\] 
Hence we can conclude that  
\[\lim_{D\rightarrow\infty}{\prod_{d=D}^{\infty}(1-\text{e}^{-d})} = 1\] 
\qed  

The examples given above of the various spaces exhibiting normal concentration of measure should convince the reader that it is real and widespread, though of course not universal. 
For the VC dimension of balls to depend linearly on $d$ is a more vague requirement as the definition of dimension for metric spaces for the purposes of similarity search is an unresolved problem. 
It is however clearly impossible to take out dimension from a discussion of the curse of dimensionality, which is after all what we have shown, caveats notwithstanding. 

This is not the first lower bound result for pivoting algorithms or indexes in general. 
There has been research for some time into lower bounds for various cases, though most often for approximate algorithms e.g. \cite{chakrabarti}. A specific lower bound for pivot-based indexing already mentioned above is that of \cite{chavez:2}: 
\[\tilde{d}\log n\] 
It is not asymptotic, and assumes that $k=\Theta(\log n)$. 
Furthermore the pivot selection is assumed to be random, as opposed to our bound that is applicable to {\em any} pivot selection technique. 

If we are to apply our restrictions on $d$ (assumed to be asymptotically equivalent to $\tilde{d}$), the result is that 
\[\tilde{d}\log n = \omega(\log^2 n)\] 
\[\tilde{d}\log n=n^{o(1)}\] 
So if we are to use these results asymptotically they do not provide strong lower bounds on the cost of similarity search. 
\chapter {Indexing experiments} 
We have borrowed basic code libraries from \cite{sisap} and followed up on the techniques and datasets in \cite{bustos} to simulate variations of the {\em incremental pivot selection} building algorithms.  
The aim of this section is to test-drive different approaches to selecting pivots to demonstrate the challenges that arise from the curse of dimensionality. 

The construction of a pivot index is essentially a matter of finding the ``right'' pivots. 
It is well known that a non-random selection of pivots can improve similarity search. 
The most commonly used technique \cite{bustos} is called incremental selection. 
The broad aim is to maximize the average distance $\rho_{p_1,\ldots,p_i}$ for each $1\leqslant i\leqslant k$. 
That is, we are doing a greedy search: at each step we find a pivot that maximizes average $\rho_{p_1,\ldots,p_i}$ and add it to our pivot list. 

To estimate the average value of $\rho_{p_1,\ldots,p_i}$, that is: 
\[\mathrm{E}_{\Omega\times\Omega}\left(\rho_{p_1,\ldots,p_i}(\omega_1,\omega_2)\right)\] 
we use a sample of $A$ pairs from $X$, and by averaging $\rho_{p_1,\ldots,p_i}(x,y)$ over all the pairs we calculate an estimate. 
In practice then we are not maximizing distance $\rho_{p_1,\ldots,p_i}$ but an $A$-sample version of it. 
Since our spaces have bounded diameter, calculating sample size in principle is not complicated, e.g. using \eqref{hoeffding}. 
The curse of dimensionality interferes however by requiring the estimation to be ever more precise for higher dimensions. 
Due to space and time limitations tradeoff between large values of $A$ and the other sample - that of pivots - must be found experimentally. 

The sample of pivots is made at each step -- $N$ random choices for the next candidate pivot are made, and the one that maximizes the ($A$-sample) distance gets added. 
Once we reach $k$ pivots, we calculate the distances from $X$ to the set of pivots. 
That distance matrix, together with the pivot list constitute the index which can be used for similarity search as already described in Algorithm \ref{alg:pivots}. 
This construction procedure is summarized in Algorithm \ref{alg:construct}. 

\begin{algorithm*}[h] 
\SetLine 
\KwData{dataset, k} 
\KwResult{pivot-list, distance-matrix} 
Sample A pairs from dataset\; 
pivot-list: \{\} \; 
\For {i in 1 to k}{ 
Select N candidates for a pivot\; 
Find pivot maximizing distance based on pairs\; 
Add selected pivot to pivot-list\; 
} 
calculate distance-matrix\; 
return pivot-list, distance-matrix\; 
\caption{constructing a pivot-based index} 
\label{alg:construct} 
\end{algorithm*} 

A single experiment consists of constructing this index for a given dataset, and then running thousands of queries to calculate the average number of distance computations required to answer a query. 
As queries with larger radii are very slow, the experimentation, as in \cite{bustos}, is restricted to queries that return on the order of 10 to 100 elements out of a typical dataset size of 100,000. 

We have expanded on incremental pivot selection procedure by selecting pairs of points from $X$ that lie close to each other. 
The underlying idea is to concentrate on points that are hard to separate because they are so close together. 

We will not explicitly analyze the costs of doing a random sample versus this ``smart'' sample since it is clearly very expensive to arrive at the smart sample: if no such list exists a priori one would need to compile it by making an index (which incidentally was our original goal) and then making small radius range or nearest neighbour queries. 
Situations are however conceivable where such a list can be obtained at low cost. 

\begin{figure*} 
\centering 
		\includegraphics[height=3in,width=4.5in]{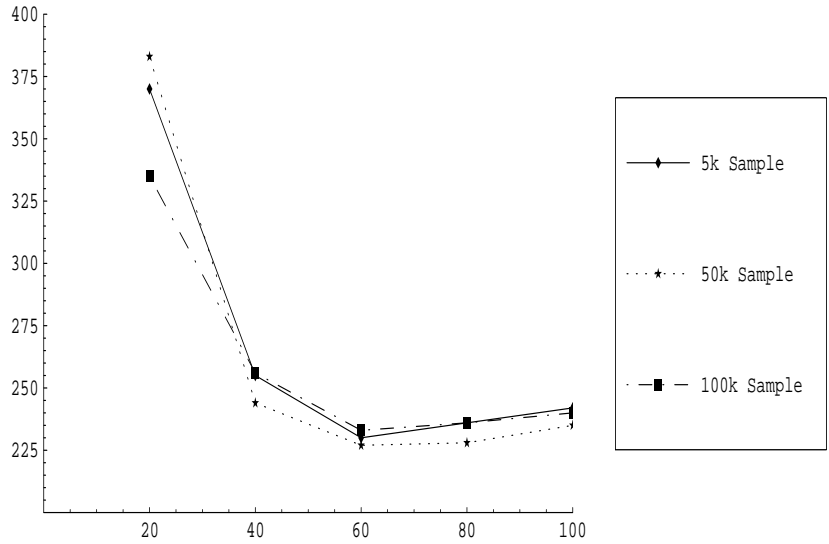} 
	\caption{\small Distance computations vs number of pivots for the NASA dataset} 
	\label{fig:resNS} 
\end{figure*} 
In Figure \ref{fig:resNS} we have plotted results for a dataset coming from compressed image data from NASA \cite{dimacs}. 
This dataset has 40000 20-dimensional vectors, equipped with the Euclidian $L^2$ metric. 
Incidentally its estimated Ch\'avez intrinsic dimension is 3.9. 
This dataset presented certain challenges as results lacked the stability of other datasets, e.g. for pivot size 20 the difference between worst and best result for the same parameters can be as high as 14 distance computations. 
This uncertainty remains even after varying the pivot sample size $N$. In our experiments we leave $N$ fixed at 40 since varying it doesn't change results either, not just their variance. 
On the other hand this difference of 14 represents only about $4\%$ so it doesn't affect the broad conclusions. 
The experiment was run using different sample sizes for the incremental selection technique: 5000, 50000, and 10000. 
Testing was done by averaging 5000 queries that return an average of 40 points (that is $0.1\%$ of the dataset). 
The overall result is that indexes based on the different samples all achieve about the same performance, with the optimal number of pivots around 60. 
The conclusion is that significant savings in construction time can be achieved by reducing sample size. 
Query time and index size however remain the same for all 3 cases. 

\begin{figure*} 
\centering 
		\includegraphics[height=3in,width=4.5in]{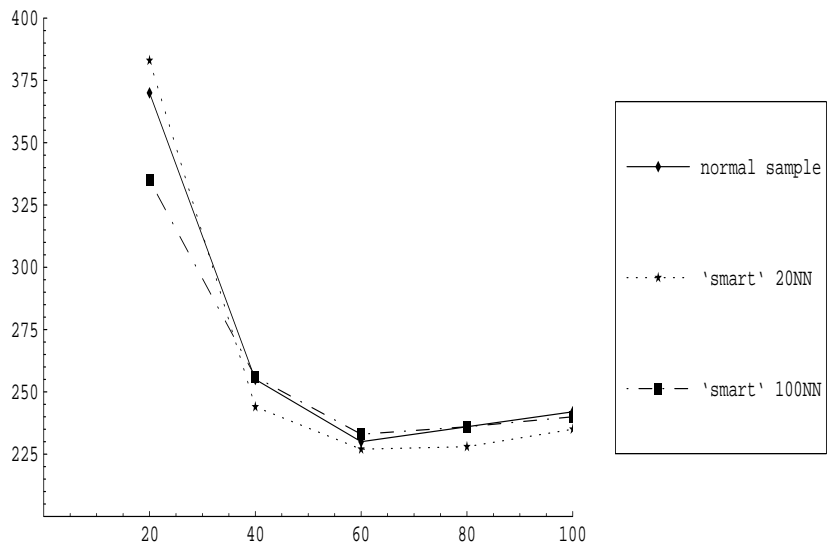} 
	\caption{\small Distance computations vs number of pivots for the NASA dataset} 
	\label{fig:resNM} 
\end{figure*} 

Figure \ref{fig:resNM} illustrates what happens when we select ``smart'' samples instead of random $A$ pairs as before. 
All sample sizes are 5000, testing done by averaging 5000 queries. 
We produce two sets of smart samples: one based on running 20-nearest-neighbour (20NN) queries, where 5000 random queries are made and their 20th nearest neighbour is selected as the second element of the pair, thus producing 5000 pairs. 
The second sample is produced similarly using 100NN queries. 
Here we notice that the choice of which smart sample to select matters somewhat, as the 20NN leads to a slightly better performing index. 
The difference is however small and must be balanced with the expense of construction for as we mentioned before we need to build a preliminary index to create the smart sample. 

\begin{figure*} 
\centering 
		\includegraphics[height=3in,width=4.5in]{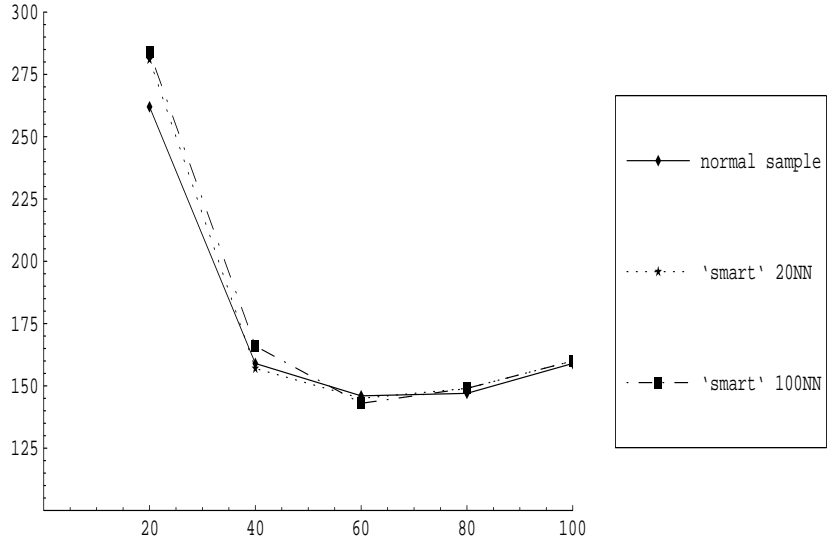} 
	\caption{\small Distance computations vs number of pivots for an 8-dimensional uniform cube dataset} 
	\label{fig:res1} 
\end{figure*} 

We repeat the smart sample experiments on a synthetic dataset: the uniform cube of dimension 8. 
The results here show even less difference between random and smart samples of $A$ pairs. 
In addition variation is very low, with repeat experiments giving virtually identical results. 
Perhaps more interesting is to do simulations where the dimension $d$ rises. 
Unfortunately hardware and software limitations severely limit the number $d$, with our setup running out of memory at $d=14$. 
We can point however to results by \cite{bustos} where performance for optimal number of pivots is plotted against various values of $d$ from 2 to 14, for uniformly distributed cubes. 
The results clearly demonstrate exponential growth in $d$ no matter the pivot selection technique. 
In other words, so far, no matter how contrived the building algorithm for pivot-based indexing we have yet to find one that breaks the curse of dimensionality. 
Our proposal to vary the sample selection for the incremental technique performs just like the rest in that respect. 

In summary, it is impossible to prove the curse of dimensionality by simulations alone as we could always ask ourselves maybe another index building algorithm exists. 
This is why our theoretical result is so powerful: we are able to make the conclusion for {\em any} pivot-index building method. 
We are leaving thus several options for the practitioner: to try better indexes, control intrinsic dimension, or limit the size of the dataset. 

The solution so far seems to have been to simply mostly avoid high-dimensional datasets. 
Bar some breakthrough in indexing technology the advice then for those dealing with truly high-dimensional data is to limit the dataset size in relation to available processing power: this means no exponential growth in size is to be allowed. 

\bibliographystyle{alpha} 
 
\end{document}